\documentclass[
final
]{dmtcs-episciences}

\usepackage[utf8]{inputenc}
\usepackage{subfigure,amssymb, latexsym, amsfonts, amsmath, lscape, amscd,
	amsthm, color, epsfig, mathrsfs, tikz, enumerate,tikz-cd,mathtools,hyperref,tabularx}

\usepackage{csquotes}
\usepackage{tikz,pgf,tkz-graph,geometry}
\usepackage{tikz-cd,multicol}
\usepackage{floatrow}
\usetikzlibrary{shapes.geometric,calc,shapes,fit,intersections,graphs,graphs.standard,decorations.pathreplacing,decorations.markings,positioning,arrows,shadows,arrows.meta,fadings,decorations.text}
\usetikzlibrary{arrows,automata}
\usetikzlibrary{graphs}
\usetikzlibrary{graphs.standard}
\tikzset{every state/.style={minimum size=0pt}}
\usepackage{color}
\usepackage{mathtools}
\usepackage{pgf,tikz}
\usetikzlibrary{decorations.pathreplacing}
\usetikzlibrary{decorations.markings}
\usetikzlibrary{positioning}
\usetikzlibrary{arrows}

\tikzset{%
glow/.style={%
preaction={#1, draw, line join=round, line width=0.5pt, opacity=0.04,
preaction={#1, draw, line join=round, line width=1.0pt, opacity=0.04,
preaction={#1, draw, line join=round, line width=1.5pt, opacity=0.04,
preaction={#1, draw, line join=round, line width=2.0pt, opacity=0.04,
preaction={#1, draw, line join=round, line width=2.5pt, opacity=0.04,
preaction={#1, draw, line join=round, line width=3.0pt, opacity=0.04,
preaction={#1, draw, line join=round, line width=3.5pt, opacity=0.04,
preaction={#1, draw, line join=round, line width=4.0pt, opacity=0.04,
preaction={#1, draw, line join=round, line width=4.5pt, opacity=0.04,
preaction={#1, draw, line join=round, line width=5.0pt, opacity=0.04,
preaction={#1, draw, line join=round, line width=5.5pt, opacity=0.04,
preaction={#1, draw, line join=round, line width=6.0pt, opacity=0.04,
}}}}}}}}}}}}}}
\newtheorem{theorem}{Theorem}[section]

\newtheorem{corollary}[theorem]{Corollary}

\newtheorem{lemma}[theorem]{Lemma}
\newtheorem{proposition}[theorem]{Proposition}

\newtheorem{remark}[theorem]{Remark}
\newtheorem{defn}[theorem]{Definition}
\usepackage[numbers]{natbib}

\author[Sen et al.]{Sagnik Sen\affiliationmark{1}\thanks{Supported by French ANR project ``HOSIGRA'' (ANR-17-CE40-0022), IFCAM project ``Applications of graph homomorphisms''
(MA/IFCAM/18/39)and SERB-MATRICS ``Oriented chromatic and clique number of planar graphs'' (MTR/2021/000858)}
  \and \'Eric Sopena\affiliationmark{2}
  \and S. Taruni\affiliationmark{1,3}\thanks{Supported by Centro de Modelamiento Matemático (CMM) BASAL fund FB210005 for center of excellence from ANID-Chile.}}
\title{Homomorphisms of $(n,m)$-graphs with respect to generalized switch}
\affiliation{
  Indian Institute of Technology Dharwad, Karnataka, India.\\
 Univ. Bordeaux, CNRS, Bordeaux INP, LaBRI, UMR5800, F-33400 Talence, France.\\
  Centro de Modelamiento Matemático (CNRS IRL2807), Universidad de Chile, Santiago, Chile.}
\begin{document}
% This is only used if you are compiling for a volume before vol 25
% \publicationdetails{VOL}{2015}{ISS}{NUM}{SUBM}
% This is the new form of collecting the data, starting with vol 25
\publicationdata{vol. 27:3}{2025}{26}{10.46298/dmtcs.13196}{2024-03-08; 2024-03-08; 2025-08-12}{2025-11-20}
% \publicationdata
% {vol. 25:3 special issue for main purpose}
% {2022}
% {1}
% {10.46298/dmtcs.10472}
%{1998-10-14; 1998-10-14; 2002-07-19; 2014-02-05; 2015-09-09; 2022-12-25}
%{2022-12-3}
% {2022-12-3; None}
% {2023-1-1}
\maketitle
\begin{abstract}
  The study of homomorphisms of $(n,m)$-graphs, that is, 
adjacency preserving vertex mappings of graphs with $n$ types of arcs and $m$ types of edges was initiated by Ne\v{s}et\v{r}il and Raspaud in 2000. 
Later, some attempts were made to generalize the switch operation that is popularly used in the study of signed graphs, and study its effect on the above mentioned homomorphism. 

In this article, we too provide a generalization of the switch operation on $(n,m)$-graphs, which to the best of our knowledge, encapsulates all the previously known generalizations as special cases. We approach the study of homomorphisms with respect to the switch operation axiomatically. We prove some fundamental results that are essential tools in the further study of this topic.
In the process of proving the fundamental results, we have provided yet another solution to an open problem posed by
Klostermeyer and MacGillivray in 2004.
We also prove the existence of a categorical product for $(n,m)$-graphs with respect to a particular class of generalized switch which implicitly uses category theory. This is a counter intuitive solution as the number of vertices in the Categorical product of two $(n,m)$-graphs on $p$ and $q$ vertices has a multiple of $pq$ many vertices, where the multiple depends on the switch. 
This solves an open question asked by Brewster in 
the PEPS 2012 workshop as a corollary. 
We also provide a way to calculate the product explicitly, and prove general properties of the product.  We define the analog of chromatic number for $(n,m)$-graphs with respect to generalized switch and explore the interrelations between chromatic numbers with respect to different switch operations. We find the value of this chromatic number for the family of forests using group theoretic notions.  
\end{abstract}

\section{Introduction}
A \textit{graph homomorphism} is an edge-preserving vertex mapping of a graph $G$ to a graph $H$. It is also known as an \textit{$H$-coloring} of $G$ and the notion was introduced as a generalization of coloring~\cite{maurer1981complexity}. It allows us to unify certain important constraint satisfaction problems, especially related to scheduling and frequency assignments, which are otherwise
expressed as various coloring and labeling problems on graphs~\cite{hell2004graphs}.
Thus, the notion of graph homomorphism manages to capture a wide range of important applications in an uniform setup. When viewed as an operation on the set of all graphs, it induces rich algebraic structures: a quasi order (and a partial order), a lattice, and a category~\cite{hell2004graphs}.

The study of graph homomorphisms can be classified into three major branches:
\begin{enumerate}[(i)]
\item  The study of various 
application-motivated optimization problems are
modeled using graph homomorphisms. These usually involve finding  an $H$ having certain prescribed properties such that every member of a graph family $\mathcal{F}$ is $H$-colorable~\cite{gutin2006level,hell1990complexity,sopena1997chromatic,zhu1996uniquely}.

\item The study of the algorithmic aspects of the $H$-coloring problem, including characterizing its dichotomy, and finding exact (polynomial) approximation or parameterized algorithms for the hard problems~\cite{brewster2017complexity,feder2003acyclic,foucaud2022graph}.

\item The study of the algebraic structures that arise from the notion of graph homomorphisms~\cite{hell2004graphs}.
\end{enumerate}

Unsurprisingly, these three areas of research have interdependencies and connections. The notion of graph homomorphisms, initially introduced for undirected and directed graphs, later got extended to hypergraphs~\cite{hahn1997graph,nevsetvril1999homomorphism}, $2$-edge-colored graphs~\cite{alon1998homomorphisms}, $k$-edge-colored graphs~\cite{guspiel2017universal} and $(n,m)$-graphs~\cite{nevsetvril2000colored}. 
These graphs, due to their various types of adjacencies, manage to capture complex relational structures and are useful for mathematical modeling. For instance, the Query Evaluation Problem (QEP) in graph database (the immensely popular databases that are now used to handle highly interrelated data in social networks like Facebook, Twitter, etc.),
is modeled on homomorphisms of $(n,m)$-graphs~\cite{angles2017foundations, beaudou2019complexity}.

From a theoretical point of view, apart from being the generalization of the 
well-studied graph homomorphisms of oriented, signed, $2$-edge-colored, and $k$-edge-colored graphs, homomorphisms of $(n,m)$-graphs are known to have connections with some topics in graph theory as well as other mathematical disciplines. On the one hand, they relate to graph theoretic notions such as 
harmonious coloring~\cite{alon1998homomorphisms},
flows~\cite{BORODIN2004147}, 
and 
acyclic coloring~\cite{nevsetvril2000colored}.
On the other hand, they also have connections with 
the study of Coxeter 
groups~\cite{alon1998homomorphisms}, and binary predicate logic~\cite{nevsetvril2000colored}. 
It is worth mentioning that in a work by 
Borodin, Kim, Kostochka, West~\cite{BORODIN2004147}, the then best approximation of 
Jaeger's conjecture for planar graphs was established 
by showing that 
every planar graph with girth 
at least $\frac{20t-2}{3}$ has circular chromatic number at most $2+\frac{1}{t}$. 
However, this follows as a corollary of 
the main result of the paper~\cite{BORODIN2004147} which is a theorem on homomorphisms of $(n,m)$-graphs.

Thus, indeed the study of homomorphisms of $(n,m)$-graphs is
a significant area of research. 
However, as there are not many known 
well-structured $(n,m)$-graphs, the study of its homomorphisms becomes extremely difficult. In contrast, there are known well-structured oriented, and signed graphs such as the Paley tournaments, signed Paley graphs, and Tromp constructions~\cite{DBLP:journals/jgt/Marshall07, ochem2014homomorphisms}. 
In an other work~\cite{lahiri2021chromatic}, we have implemented the theory from this work (related to the switch operation) to construct some well-structured $(0,3)$ and $(1,1)$-graphs and used them to prove upper bounds for $(n,m)$-chromatic number of partial $2$-trees.  

In recent times, researchers have started further extending the graph homomorphism studies by exploring the effect of switch operation on homomorphisms. 
Notably, homomorphisms of 
signed graphs, which are essentially obtained by observing the effect of the switch operation on $2$-edge-colored graphs, has gained immense popularity~\cite{naserasr2015homomorphisms, naserasr2021homomorphisms,ochem2014homomorphisms,sopena2013homomorphisms} due to its strong connection with graph minor theory. 
Also, graph homomorphism with respect to some other switch-like operations, such as, 
push operation on oriented graphs~\cite{klostermeyer2004homomorphisms}, 
cyclic switch on $k$-edge-colored graphs~\cite{brewster2009edge}, 
and switching $(n,m)$-graphs with respect to Abelian groups of special type (which does not allow switching an edge to an arc or vice versa)~\cite{brewster20222,leclerc2021switching} and switching $(n,m)$-graphs with respect to non-Abelian group~\cite{duffy2022switching} to list a few, has been recently studied.  Naturally, all three main branches of research listed above in the context of graph homomorphism are also explored for the above-mentioned extensions and variants. However, in comparison, the global algebraic structure is a less explored branch for the extensions. 

\medskip

\noindent \textbf{Organization:}  In this article, we introduce a homomorphism with respect to a generalized switch 
operation (using a group $\Gamma$) on $(n,m)$-graphs which, in particular, allows arcs to become edges and vice-versa. 
 We call it 
a $\Gamma$-homomorphism whose detailed definition is deferred to the next section. 
In particular, 
it is possible to view the set of all $(n,m)$-graphs as a category with $\Gamma$-homomorphism 
playing the role of morphism. However, in this article, we have refrained from using the language of category theory as much as possible, and have used the language of graph theory instead.

\begin{itemize}
    \item In Section \ref{Sec hom} we introduce the  notion of $\Gamma$-homomorphisms of $(n,m)$-graphs  and prove some basic results. In particular, we show that
the switch operation defined by us is (strictly) more general than the switch operation defined by 
Leclerc, MacGillivray, and Warren~\cite{leclerc2021switching}. 

\item  In Section \ref{Sec alge prop} we study algebraic 
properties of $\Gamma$-homomorphisms and explore 
their relation with  $(n,m)$-homomorphisms. 
In particular, we solve a generalized version of an open problem due to Klostermeyer and MacGillivray~\cite{klostermeyer2004homomorphisms}. Also, we show that the important notion of ``core'' is well-defined in this setup (category).

\item In Section \ref{sec category} we establish the existence of a categorical product and a co-product of $(n,m)$-graphs under $\Gamma$-homomorphism. Furthermore, we prove some fundamental properties of these categorical products and co-products.  
Interestingly, 
given two $(n,m)$-graphs $G$ and $H$ of order $p$ and $q$, respectively, their categorical product is a graph on $|\Gamma|pq$  vertices. 

\item In Section \ref{sec chromatic number} we define and study the $\Gamma$-chromatic number for the family of forests, where the $\Gamma$-chromatic number is defined using $\Gamma$-homomorphism. This result generalizes Theorem~1.1 of~\cite{nevsetvril2000colored}.

\item In Section \ref{sec conclusion} we conclude our work and propose possible future directions for research on this topic. 
\end{itemize}

\section{Homomorphisms of $(n,m)$-graphs and generalized switch} \label{Sec hom}
Throughout this article, we follow the standard graph theoretic, algebraic and category theory notions from West~\cite{west2001introduction},  Artin~\cite{Artin:1998}, and Jacobson~\cite{jacobson2012basic}, respectively.

  An \textit{$(n,m)$-graph} $G$ is a graph with vertex set $V(G)$, arc set $A(G)$ and edge set $E(G)$, where each arc is colored with one of the $n$ colors from $\{2, 4,  \ldots, 2n\}$ and each edge is colored with one of the $m$ colors from $\{2n+1, 2n+2, \ldots, 2n+m\}$. See Figure~\ref{Fig: (n,m)-graph}\textcolor{red}{(a)} for an example of a $(2,2)$-graph.
In particular, if there is an arc of color $i$ from $u$ to $v$, we, equivalently, view it as a \textit{reverse arc} from $v$ to $u$ of color $(i-1)$. In this scenario, 
we say that $v$ is an $i$-neighbor of $u$, or equivalently, 
$u$ is an $(i-1)$-neighbor of $v$. Furthermore, if there is an edge of color $j$ between $u$ and $v$, then we say that $u$ (resp., $v$) is a $j$-neighbor of $v$ (resp., $u$). Figure~\ref{Fig: (n,m)-graph}\textcolor{red}{(b)} depicts all possible types of adjacencies of the vertex $u$ in a $(2,2)$-graph.
For convenience, we use the following convention throughout this article: if $u$ is an $i$-neighbor of $v$, then we say $v$ is an $\overline{i}$-neighbor of $u$. In particular, $\overline{\overline{i}} = i$.

\begin{figure}[htp]
\begin{tabularx}{\textwidth}{X  X}
\centering
\begin{tikzpicture}[scale = 0.8, inner sep=0.5mm, auto=center]	
		\begin{scope}[very thick,decoration={markings,
									mark=at position 0.5 with {\arrow{>}}}
								] 

    \node[circle,draw] (a) at (0,0) {};
    \node[circle,draw] (b) at (0,3) {};
    \node[circle,draw] (c) at (1,1.5) {};
    \node[circle,draw] (d) at (2,0) {};
    \node[circle,draw] (e) at (2,3) {};
    \node[circle,draw] (f) at (3,1.5) {};
    \node[circle,draw] (g) at (4,0) {};
    \node[circle,draw] (h) at (4,3) {};

    \draw[red] (a)to(b);
    \draw[-latex,line width= 1pt] (b) to (c);
    \draw[blue] (c)to(a);
    \draw[-latex,line width = 1pt] (f)to(g);
     \draw[-latex,line width = 1pt] (f)to(c);
     \draw[-latex,line width = 1pt] (g)to(h);
     \draw[-latex,line width = 1pt] (h)to(f);
    \draw[-latex, green, line width = 1pt] (e)to(c);
    \draw[-latex, green, line width = 1pt] (d)to(c);
    \draw[blue] (e)to(f)to(d);
    	
                            \end{scope}

\end{tikzpicture}

 \label{Fig: mixed graph}

&
\centering
\begin{tikzpicture}[scale=0.5, inner sep=0.5mm, auto=center]	
		\begin{scope}[very thick,decoration={markings,
									mark=at position 0.5 with {\arrow{>}}}
								] 
    	\node[circle,draw] (c1) at (0,0)  {u};
							\node[circle,draw] (c2) at (0.5*360/6:4cm)  {};
							\node[circle,draw] (c3) at (1.1*360/6:4cm)  {};
							\node[circle,draw] (c4) at (1.9*360/6:4cm)  {};
							\node[circle,draw] (c5) at (2.5*360/6:4cm)  {};
							\node[circle,draw] (c6) at (3*360/6:4cm)  {};
							\node[circle,draw] (c7) at (6*360/6:4cm)  {};
							\draw[-latex,line width = 1pt] (c1) -- (c2);
                            \node at (0.15*360/6:2cm) {\scriptsize 1};
                            \node at (0.65*360/6:2cm) {\scriptsize 2};
                             \node at (1.25*360/6:2cm) {\scriptsize 3};
                              \node at (1.75*360/6:2cm) {\scriptsize 4};
                               \node at (2.35*360/6:2cm) {\scriptsize 5};
                                \node at (2.85*360/6:2cm) {\scriptsize 6};
							\draw[-latex, green, line width = 1pt] (c3) -- (c1);
							\draw[-latex, green, line width = 1pt] (c1) -- (c4);
							\draw[blue] (c1) -- (c5);
							\draw[red] (c6) -- (c1);
							\draw[-latex,line width = 1pt] (c7) -- (c1);
                            \end{scope}
\end{tikzpicture}

 \label{Fig: mixed graph2}
\end{tabularx}
\caption{(a) An example of a $(2,2)$-graph $G$.
(b) The $2n+m = 6$ possible types of adjacencies for a vertex $u$ in a $(2,2)$-graph.}
\label{Fig: (n,m)-graph}
\end{figure}

Let $\Gamma \subseteq S_{2n+m}$, where $S_{2n+m}$ is the permutation group on 
$A_{n,m} = \{1, 2, \ldots, 2n, 2n+1, \ldots, 2n+m\}$. A \textit{$\sigma$-switch} at a vertex $v$ of an $(n,m)$-graph is to change the incident arcs and edges of $v$ in such a way that its $t$-neighbors become $\sigma(t)$-neighbors for all $t \in A_{n,m}$ where $\sigma \in \Gamma$. 
To \textit{$\Gamma$-switch}  a vertex $v$ of an $(n,m)$-graph is to 
apply a $\sigma$-switch on $v$ for some $\sigma \in \Gamma$. An $(n,m)$-graph $G'$ obtained by a sequence of $\Gamma$-switches performed on the vertices of $G$ is a \textit{$\Gamma$-equivalent} graph of $G$. For an example, consider the $\Gamma = \langle \sigma \rangle$, a subgroup of $S_6$ that is generated by $
  \sigma =
  \left(
  \begin{matrix}
    1 & 2 & 3 & 4 & 5 & 6 
  \end{matrix}
  \right)$. Figure~\ref{Fig: gamma equivalent graph} is a $(1,2)$-graph $G'$ that is $\Gamma$-equivalent graph of $G$ from the example shown in Figure~\ref{Fig: (n,m)-graph}\textcolor{red}{(a)} which is obtained after switching the highlighted vertices using $\sigma$-switch. 
  Notice that, the switch operation, the way it is defined, may freely convert an arc (resp., reverse arc, edge) of any color to an arc, reverse arc, or an edge of any color.  
  Thus, the set of arcs and edges of $G$ and $G'$ may differ even if $G$ and $G'$ are $\Gamma$-equivalent.

  \begin{figure}[htp]
\begin{tabularx}{\textwidth}{X  X}
\centering
\begin{tikzpicture}[scale = 0.8, inner sep=0.5mm, auto=center]	
		\begin{scope}[very thick,decoration={markings,
									mark=at position 0.5 with {\arrow{>}}}
								] 

    \node[circle,draw] (a) at (0,0) {};
    \node[circle,draw] (b) at (0,3) {};
    \node[circle,draw] (c) at (1,1.5) {};
    \filldraw[orange!40] (2,0) circle (12pt);
    \node[circle,draw] (d) at (2,0) {};
    \filldraw[orange!40] (2,3) circle (12pt);
    \node[circle,draw] (e) at (2,3) {};
    
    \node[circle,draw] (f) at (3,1.5) {};
    \node[circle,draw] (g) at (4,0) {};
    \node[circle,draw] (h) at (4,3) {};

    \draw[red] (a)to(b);
    \draw[-latex,line width= 1pt] (b) to (c);
    \draw[blue] (c)to(a);
    \draw[-latex,line width = 1pt] (f)to(g);
     \draw[-latex,line width = 1pt] (f)to(c);
     \draw[-latex,line width = 1pt] (g)to(h);
     \draw[-latex,line width = 1pt] (h)to(f);
    \draw[-latex, green, line width = 1pt] (e)to(c);
    \draw[-latex, green, line width = 1pt] (d)to(c);
    \draw[blue] (e)to(f)to(d);
    	
                            \end{scope}

\end{tikzpicture}

&
\centering

\begin{tikzpicture}[scale = 0.8, inner sep=0.5mm, auto=center]	
		\begin{scope}[very thick,decoration={markings,
									mark=at position 0.5 with {\arrow{>}}}
								] 

    \node[circle,draw] (a) at (0,0) {};
    \node[circle,draw] (b) at (0,3) {};
    \node[circle,draw] (c) at (1,1.5) {};
    \node[circle,draw] (d) at (2,0) {};
    \node[circle,draw] (e) at (2,3) {};
    \node[circle,draw] (f) at (3,1.5) {};
    \node[circle,draw] (g) at (4,0) {};
    \node[circle,draw] (h) at (4,3) {};

    \draw[red] (a)to(b);
    \draw[-latex,line width= 1pt] (b) to (c);
    \draw[blue] (c)to(a);
    \draw[-latex,line width = 1pt] (f)to(g);
     \draw[-latex,line width = 1pt] (f)to(c);
     \draw[-latex,line width = 1pt] (g)to(h);
     \draw[-latex,line width = 1pt] (h)to(f);
    \draw[blue] (e)to(c);
    \draw[blue] (d)to(c);
    \draw[red] (e)to(f)to(d);
    	
                            \end{scope}

\end{tikzpicture}
\end{tabularx}
\caption{A $(2,2)$-graph $G$ and its $\Gamma$-equivalent graph $G'$.}
\label{Fig: gamma equivalent graph}
\end{figure}

In the very first work on $(n,m)$-graphs, Ne\v{s}et\v{r}il and Raspaud~\cite{nevsetvril2000colored} in $2000$, extended the notion of graph homomorphisms to homomorphisms of $(n,m)$-graphs\footnote{In their work, $(n,m)$-graphs were termed as colored mixed graphs and $(n,m)$-homomorphisms of $(n,m)$-graphs as colored homomorphism of colored mixed graphs. Also, in~\cite{BORODIN2004147}, $(n,m)$-graphs are referred as $s$-graphs.}, this generalization for particular cases implies the study related to homomorphisms of oriented, signed and $k$-edge colored graphs~\cite{alon1998homomorphisms, brewster2009edge, naserasr2015homomorphisms,naserasr2021homomorphisms, sopena1997chromatic, sopena2016homomorphisms}. 

\begin{defn}\label{def hom}
	Let $G$ and $H$ be two $(n,m)$-graphs. An $(n,m)$-homomorphism of $G$ to $H$ is a vertex mapping $\phi\colon V(G) \to V(H)$ satisfying the following: for any $u,v \in V(G),$  if $u$ is a $t$-neighbor of $v$, then $f(u)$ is a $t$-neighbor of $f(v)$ in $H$, where $t \in A_{n,m}$.
\end{defn}

We now extend the Definition \ref{def hom} to homomorphisms of $(n,m)$-graphs with $\Gamma$-switch. 

\begin{defn}

A \textit{$\Gamma$-homomorphism} of  $G$ to $H$ is a function $f\colon V(G) \to V(H)$ such that there exists a  $\Gamma$-equivalent graph $G'$ of $G$ satisfying the following: for any $u,v \in V(G) = V(G')$, if $u$ is a $t$-neighbor of $v$ in $G'$ 
then $f(u)$ is a $t$-neighbor of $f(v)$ in $H$, where $t \in A_{n,m}$. We denote this by $G \xrightarrow{\Gamma} H$.

\end{defn}

 \begin{figure}[htp]
\begin{tabularx}{\textwidth}{X  X}
\centering
\begin{tikzpicture}[scale = 0.6, inner sep=0.5mm, auto=center]	
		\begin{scope}[very thick,decoration={markings,
									mark=at position 0.5 with {\arrow{>}}}
								] 

    \node[circle,draw] (a) at (0,0) {};
    \node[circle,draw] (b) at (0,3) {};
    \node[circle,draw] (c) at (1,1.5) {};
   \filldraw[orange!40] (2,0) circle (12pt);
    \node[circle,draw] (d) at (2,0) {};
   \filldraw[orange!40] (2,3) circle (12pt);
    \node[circle,draw] (e) at (2,3) {};
    \node at (2,-1) {$G$};
    
    \node[circle,draw] (f) at (3,1.5) {};
    \node[circle,draw] (g) at (4,0) {};
    \node[circle,draw] (h) at (4,3) {};

    \draw[red] (a)to(b);
    \draw[-latex,line width= 1pt] (b) to (c);
    \draw[blue] (c)to(a);
    \draw[-latex,line width = 1pt] (f)to(g);
     \draw[-latex,line width = 1pt] (f)to(c);
     \draw[-latex,line width = 1pt] (g)to(h);
     \draw[-latex,line width = 1pt] (h)to(f);
    \draw[-latex, green, line width = 1pt] (e)to(c);
    \draw[-latex, green, line width = 1pt] (d)to(c);
    \draw[blue] (e)to(f)to(d);
    	
                            \end{scope}

\end{tikzpicture}

&

\centering

\begin{tikzpicture}[scale = 0.6, inner sep=0.5mm, auto=center]	
		\begin{scope}[very thick,decoration={markings,
									mark=at position 0.5 with {\arrow{>}}}
								]

    \node[circle,draw] (a) at (0,0) {};
    \node at (0,-0.5) {a};
    \node[circle,draw] (b) at (0,3) {};
    \node at (0,3.5) {b};
    \node[circle,draw] (g) at (4,0) {};
     \node at (4,3.5) {c};
    \node[circle,draw] (h) at (4,3) {};
    \node at (4,-0.5) {d};
    \node at (2,-1) {$H$};
    \draw[red] (h)to(b);
    \draw[-latex,line width= 1pt] (a) to (b);
    \draw[blue] (h)to(g);
    \draw[-latex,line width = 1pt] (g)to(a);
     \draw[-latex,line width = 1pt] (b)to(g);

                            \end{scope}

\end{tikzpicture}
\end{tabularx}
\vspace{0.2cm}
\centering

\begin{tikzpicture}[scale = 0.6, inner sep=0.5mm, auto=center]	
		\begin{scope}[very thick,decoration={markings,
									mark=at position 0.5 with {\arrow{>}}}
								] 
        \draw [{Latex[length=1.5mm]}-] (-1,1.5)  to [bend left=30] node [below, sloped] (-2.7,2) {$\sigma$-switch} (-3,4.5);                        

    \node[circle,draw] (a) at (0,0) {};
     \node at (0,-0.5) {c};
    \node[circle,draw] (b) at (0,3) {};
      \node at (0,3.5) {b};
    \node[circle,draw] (c) at (1,1.5) {};
      \node at (1,2.15) {d};
    \node[circle,draw] (d) at (2,0) {};
    \node at (2,-0.5) {c};
    \node[circle,draw] (e) at (2,3) {};
     \node at (2,3.5) {c};
    \node[circle,draw] (f) at (3,1.5) {};
     \node at (3,2.15) {b};
    \node[circle,draw] (g) at (4,0) {};
     \node at (4,-0.5) {d};
    \node[circle,draw] (h) at (4,3) {};
     \node at (4,3.5) {a};
     \node at (2,-1.5) {$G'$};

    \draw[red] (a)to(b);
    \draw[-latex,line width= 1pt] (b) to (c);
    \draw[blue] (c)to(a);
    \draw[-latex,line width = 1pt] (f)to(g);
     \draw[-latex,line width = 1pt] (f)to(c);
     \draw[-latex,line width = 1pt] (g)to(h);
     \draw[-latex,line width = 1pt] (h)to(f);
    \draw[blue] (e)to(c);
    \draw[blue] (d)to(c);
    \draw[red] (e)to(f)to(d);
    \draw [{Latex[length=1.5mm]}-] (8,4.5)  to [bend left=30] node [below, sloped] (7.7,2) {\small $\langle e \rangle$-homomorphism} (6,1.5); 
    	
                            \end{scope}

\end{tikzpicture}

\caption{An example of a $\Gamma$ homomorphism: $G \xrightarrow{\Gamma} H$.}
\label{Fig: gamma homomorphism}
\end{figure}

A \textit{$\Gamma$-isomorphism} 
of $G$ to $H$ is a bijective $\Gamma$-homomorphism whose inverse is also a $\Gamma$-homomorphism.  We denote this by $G \equiv_{\Gamma} H$. Observe that if $\Gamma = \langle e \rangle$ is the singleton group with the identity element $e$, then our $\Gamma$-homomorphism definition becomes the same as homomorphism of $(n,m)$-graphs.

 A related work~\cite{leclerc2021switching} on homomorphism with respect to a switch operation on $(n,m)$-graphs has been studied in which, an Abelian subgroup of $S_m \otimes (S_2 \wr S_n)$ acts on the vertices of $(n,m)$-graph such that the incident edges switch color with edges and the incident arcs switch color with arcs. Formally, let $\phi \in S_m, \psi \in S_n, $ and $\pi = (p_1, p_2, \ldots ,p_n) \in (\mathbb{Z}_2)^{n}$, for an ordered triple $\gamma = (\phi, \psi, \pi),$ an \textit{LMW-switch}\footnote{No particular name was given for this switch in~\cite{leclerc2021switching}. We use the initials of the author names for convenience.}~\cite{leclerc2021switching} at a vertex $v$ is said to be the process of transforming $G$ into an $(n,m)$-graph $G^{v,\gamma}$ where edges with color $i$ incident to $v$  changes to edges with color $\phi(i)$, arcs of color $j$ incident to $v$, changes to arcs of color $\psi(j)$ with the orientation reversed if and only if $p_j = 1$ in $\pi$. This definition is a natural extension to the definitions given in \cite{brewster2009edge, klostermeyer2004homomorphisms} for switching or pushing in the case of $(n,m) = \{(1,0), (0,m)\}$ graphs. In this paper, we study a more generalized switch operation on $(n,m)$ graphs which also captures \textit{LMW-switch} in particular.  
There have also been studies of $\Gamma$ switch-homomorphisms of $(n,m)$-graphs when $\Gamma$ is non-Abelian~\cite{kidner2021gamma}. Here, we restrict ourselves to Abelian subgroups $\Gamma$ of $S_{2n+m}$ unless otherwise stated.

\medskip

Let $u,v$ be any two adjacent vertices of an $(n,m)$-graph $G$. A \textit{$(n,m)$-switch-commutative} group $\Gamma \subseteq S_{2n+m}$ is such that 
for any $\sigma, \sigma' \in \Gamma$, by performing $\sigma$-switch on $u$ and $\sigma'$-switch on $v$, the adjacency between $u$ and $v$ changes in the 
same way irrespective of the order of the switches. 
Observe that, this property not only depends on the group but also depends on the values of $n$ and $m$. Thus, in the definition we have included the term $(n,m)$ as well. However, 
for convenience, whenever $(n,m)$ is clear from the context, we will use the term 
\textit{switch-commutative} instead of $(n,m)$-switch-commutative. In particular, whenever we are in the context of a switch-commutative group, the sequence in which we apply the switches becomes redundant. 
A $\Gamma$-switch, where $\Gamma$ is a switch-commutative group (in the context), in general, is called a \textit{commutative switch}.

\begin{theorem}~\label{Thm: lmw switch}
	Every \textit{LMW-switch} operation is a \textit{commutative switch} operation. Moreover, there exists infinitely many commutative switches
 which are not LMW-switches. 
\end{theorem}

\begin{proof}
	Let $G$ be an $(n,m)$-graph and let $u,v \in V(G)$ such that $v$ is a $t$-neighbor of $u$. Consider an Abelian subgroup  $\Gamma \subseteq S_m \otimes (S_2 \wr S_n)$. Let $\sigma_u$ and $\sigma_v$ be two LMW-switches of $\Gamma$ acting on vertices $u$ and $v$ respectively. By definition, $\sigma_u(t) = (\phi_u(t), \psi_u(t), \pi_u(t))$ is applied on $u$ and $\sigma_v(t) = (\phi_v(t), \psi_v(t), \pi_v(t))$ is applied on $v$. It is enough to prove that this LMW-switch action is a commutative switch operation. That is, to show that, the adjacency between $u$ and $v$ changes in the same way irrespective of the order in which the switches $\sigma_u$ and $\sigma_v$ are applied on $u$ and $v$, respectively. 

  Suppose $t$ is the color of an edge. As $\Gamma$ is Abelian, $\sigma_u(\sigma_v(t)) = \sigma_v(\sigma_u(t))$. As $\phi_u$ and $\phi_v$ are responsible for changing the color of the edge, which is a symmetric relation between $u$ and $v$, they commute.

 Suppose $t$ is the color of an arc. Then after applying the switches, the color of the arc changes as per the functions $\psi_u$ and $\psi_v$ and it does not matter in which order the switches are applied as the color of the arc (not the direction) is 
 a symmetric relation between $u$ and $v$, and as $\Gamma$ is commutative. On the other hand, the change in the direction of the arc is determined by $\pi_u(t) \cdot \pi_v(t)$. 
 To be precise, if  $\pi_u(t) \cdot \pi_v(t) = 0$ then the direction does not change, and if
$\pi_u(t) \cdot \pi_v(t) = 1$ then the direction changes. 
This completes the first part of the proof. 

\medskip

To prove  the moreover part, we give an example of a commutative switch which is not an $LMW$-switch. Consider $(1,1)$-graph $G$. Let $C_3 = \{ e, \sigma, \sigma^{2} \} \subseteq S_{3}$, where, 
	 $\sigma =
  \left(
  \begin{matrix}
    1 & 2 & 3  
  \end{matrix}
  \right)$.
For two adjacent vertices say $u,v \in V(G)$, $\sigma$ applied on $u$ and then $\sigma$ applied on $v$ yields the same result as that of $\sigma$ applied on $v$ first and then on $u$. Thus, $\Gamma$ is a commutative switch whereas it is clearly not a \textit{LMW-switch} as an arc (color $2$) is switched to edge (color $3$) under this operation. Further, we can extend this example to $(n,n)$-graph, for any $n \in \mathbb{N}$ and $\Gamma \subsetneq S_{3n}$, where $\Gamma = C_{3} \oplus C_3 \oplus C_3 \cdots \oplus  C_3$, we have $n$-types of arc and $n$-types of edges, where $i$-th $C_3$ is $\{ e, \sigma_i, \sigma_i^{2} \}$,  $\sigma_i =
  \left(
  \begin{matrix}
    2i-1 & 2i & 3i  
  \end{matrix}
  \right)$, respectively.  
\end{proof}

\section{Basic algebraic properties}\label{Sec alge prop}
Let $\Gamma \subseteq S_{2n+m}$ be a
switch-commutative group and let $G$ be an $(n,m)$-graph with set of vertices $\{v_1, v_2, \ldots,$ $ v_k\}$. We construct the $(n,m)$-graph $G^*$ of $G$ with respect to $\Gamma$ as follows: take $|\Gamma|$ many copies of $G$, indexed by the elements of $\Gamma$. That is, for each $\sigma \in \Gamma$, its corresponding copy of $G$
is denoted as $G^{\sigma}$. The vertex corresponding to $v_i \in V(G)$ in $G^{\sigma}$ is denoted as $v_i^{\sigma}$. A vertex $v_i^{\sigma}$ is a $t$-neighbor of $v_j^{\sigma'}$ in $G^*$ if and only if $v_i$ is a $t$-neighbor of $v_j$ in $G$ where $i,j \in \{1, 2, \ldots, k\}$ and $\sigma, \sigma' \in \Gamma$. We now define the \textit{$\Gamma$-switched graph} denoted by 
$\rho_{\Gamma}(G)$ of $G$. The $\Gamma$-switched graph $\rho_{\Gamma}(G)$ on $(|\Gamma| \times |V(G)|)$ vertices is obtained from $G^{*}$ by performing $\sigma$-switch on all the vertices of $G^{\sigma}$ for all $\sigma \in \Gamma$. 

Given an $(n,m)$-graph $G$ and two vertices $u,v \in V(G)$ such that $v$ is a $t$-neighbor of $u$, in $\rho_{\Gamma}(G)$, by definition, $v^{\sigma'}$ is either a $\overline{\sigma'(\overline{\sigma(t)})}$-neighbor of $u^{\sigma}$, or a $\sigma(\overline{\sigma'(\overline{t})})$-neighbor of $u^{\sigma}$. 
Observe that, $\overline{\sigma'(\overline{\sigma(t)})}$ and 
$\sigma(\overline{\sigma'(\overline{t})})$ may not be equal unless $\Gamma$ is a switch-commutative group. 
Thus, it is important to assume that 
$\Gamma$ is switch-commutative while defining $\rho_{\Gamma}(G)$.
Note that, the notion of $\rho_{\Gamma}(G)$ is a natural generalization of the notions push-digraph for oriented graphs~\cite{klostermeyer2004homomorphisms} or double switching graphs for signed graphs~\cite{naserasr2021homomorphisms} (in both cases the relevant groups are switch-commutative).

From the example
in Theorem~\ref{Thm: lmw switch}, we have seen that the subgroup 
$\Gamma = \langle \sigma \rangle \subsetneq S_3$, 
where $\sigma =
  \left(
  \begin{matrix}
    1 & 2 & 3  
  \end{matrix}
  \right)$,
  is a $(1,1)$-switch-commutative group. Refer to Figure~\ref{Fig: rhogamma} for the construction of $\rho_{\Gamma}(G)$ for an $(1,1)$-graph $G$. 
  
\begin{figure}
\begin{tabularx}{\textwidth}{X    X}
\centering
\begin{tikzpicture}[auto=centre, scale=0.8, inner sep=0.5mm]		                       
      \draw[fill=orange!20](3, 3) circle (1.8cm);
      \node (a) at (3.25,3.7) {$a^{\sigma^2}$};
       \node (b) at (2.25,2.75) {$b^{\sigma^2}$};
        \node (c) at (3.8,2.9) {$c^{\sigma^2}$};
        \node at (2.95,2.2) {\scriptsize \textcolor{red}{$\sigma^{2}$-switch
        applied}};
        \draw[blue,thick]  (3.5,2.8)to(3,3.5);
        \draw[-latex,line width = 1pt] (2.5,2.8) to (3,3.5);
        \draw[-latex,line width = 1pt] (3.5,2.8) to (2.5,2.8);
        \draw[glow=cyan] (3,1.5) to (-0.5,0.5);
        \draw[glow=green] (-3,1.5) to (0.5,0.5);
        \draw[glow=orange] (-2,3) to (2,3);
      
      \node at (4.2,1.1) {$G^{\sigma^2}$};
       \draw(-3, 3) circle (1.8cm);
       \node (a) at (-3,3.7) {$a$};
       \node (b) at (-2.35,2.8) {$c$};
        \node (c) at (-3.7,2.8) {$b$};
        \node at (-2.9,2.3) {\scriptsize \textcolor{red}{$e$-switch applied}};
        \draw[-latex,line width = 1pt] (-3.5,2.8)to(-3,3.5);
        \draw[blue,thick] (-2.5,2.8)to(-3,3.5);
        \draw[-latex,line width = 1pt] (-2.5,2.8) to (-3.5,2.8) ;
       \node at (-4.2,1.1) {$G$};
        \draw[fill=green!20](0, 0) circle (1.8cm);
        \node (a) at (0.1,0.7) {$a^{\sigma}$};
       \node (b) at (0.7,-0.2) {$c^{\sigma}$};
        \node (c) at (-0.7,-0.2) {$b^{\sigma}$};
        \node at (0,-0.7) {\scriptsize \textcolor{red}{$\sigma$-switch applied}};
        \draw[-latex,line width = 1pt] (-0.5,-0.2)to(0,0.5);
        \draw[blue,thick] (0,0.5)to(0.5,-0.2);
        \draw[-latex,line width = 1pt] (0.5,-0.2)to(-0.5,-0.2);
         \node at (-2.1,-1) {$G^{\sigma}$};
\end{tikzpicture}
 
&
\hspace{2cm}
\begin{tikzpicture}[scale=0.5, inner sep=0.5mm, auto=center]	
		\begin{scope}[very thick,decoration={markings,
									mark=at position 0.5 with {\arrow{>}}}
								] 
                               
	\node[circle,draw,minimum size=0.8cm] (b) at (1*360/9:5cm)  {$b$};
    \node[circle,draw,minimum size=0.8cm] (a) at (2*360/9:5cm)  {$a$};
    \node[circle,draw,minimum size=0.8cm] (c) at (3*360/9:5cm)  {$c$};
    \node[circle,draw,minimum size=0.8cm] (bsig) at (4*360/9:5cm)  {$b^{\sigma}$};
    \node[circle,draw,minimum size=0.8cm] (asig) at (5*360/9:5cm)  {$a^{\sigma}$};
    \node[circle,draw,minimum size=0.8cm] (csig) at (6*360/9:5cm)  {$c^{\sigma}$};
    \node[circle,draw,minimum size=0.8cm] (bsigsq) at (7*360/9:5cm)  {$b^{\sigma^2}$};
    \node[circle,draw,minimum size=0.8cm] (asigsq) at (8*360/9:5cm)  {$a^{\sigma^2}$};
    \node[circle,draw,minimum size=0.8cm] (csigsq) at (9*360/9:5cm)  {$c^{\sigma^2}$};
    \draw[-latex,line width = 1pt] (b)to(a);
    \draw[-latex,line width = 1pt] (c)to(b);
    \draw[-latex,line width = 1pt] (asig)to(b);
    \draw[-latex,line width = 1pt] (b)to(csigsq);
    \draw[-latex,line width = 1pt] (a)to(csig);
    \draw[-latex,line width = 1pt] (a)to(bsigsq);
    \draw[-latex,line width = 1pt] (csigsq)to(a);
    \draw[-latex,line width = 1pt] (c)to(asig);
    \draw[-latex,line width = 1pt] (asigsq)to(c);
    \draw[-latex,line width = 1pt] (bsig)to(c);
    \draw[-latex,line width = 1pt] (asigsq)to(bsig);
    \draw[-latex,line width = 1pt] (bsig)to(asig);
    \draw[-latex,line width = 1pt] (csig)to(bsig);
    \draw[-latex,line width = 1pt] (bsigsq)to(csig);
    \draw[-latex,line width = 1pt] (csig)to(asigsq);
    \draw[-latex,line width = 1pt] (csigsq)to(bsigsq);
    \draw[-latex,line width = 1pt] (bsigsq)to(asigsq);
    \draw[-latex,line width = 1pt] (asig)to(csigsq);
    \draw[blue] (b)to(csig)to(asig)to(bsigsq)to(c)to(a)to(bsig)to(csigsq)to(asigsq)to(b);
    \end{scope}
\end{tikzpicture}

\end{tabularx}

\caption{(a) The process of constructing $\rho_{\Gamma}(G)$ from the $(1,1)$-graph $G$ (also depicted). (b) The graph $\rho_{\Gamma}(G)$ where $\Gamma = \langle (1~2~3) \rangle$. Notice that $(1~2~3)$ is not a LMW-switch, but is a commutative switch.}
\label{Fig: rhogamma}
\end{figure}

% Let $G^*$ be the $(n,m)$-graph having vertices of the type 
% $v_i^{\sigma}$ where $i \in \{1, 2, \cdots, k\}$ and 
% $\sigma \in \Gamma$. Also a vertex $v_i^{\sigma}$ is a $t$-neighbor of $v_j^{\sigma'}$ in $G^*$ if and only if $v_i$ is a $t$-neighbor of $v_j$ in $G$ where $i,j \in \{1, 2, \cdots, k\}$ and $\sigma, \sigma' \in \Gamma$. The \textit{$\Gamma$-switched graph} denoted by 
% $\rho_{\Gamma}(G)$ of $G$ is the $(n,m)$-graph obtained from $G^*$
% by performing a $\sigma$-switch on $v_i^{\sigma}$ for all $i \in \{1, 2, \cdots, k\}$ and $\sigma \in \Gamma$.

\begin{lemma}\label{lem convention}
	Let $u,v$ be two vertices of an $(n,m)$-graph $G$ such that $v$ is a $t$-neighbor of $u$.  Suppose $\sigma_u, \sigma_v$ are applied on $u$ and $v$ respectively. Then $\Gamma$ is a switch-commutative group if any only if we have $\overline{\sigma_v(\overline{\sigma_u(t)})} = \sigma_u(\overline{\sigma_v(\overline{t})})$, for all $\sigma_u, \sigma_v \in \Gamma$ and for all $t \in A_{n,m}.$
\end{lemma}

\begin{proof}
Let $u$ and $v$ be vertices of $G$ such that $v$ is a $t$-neighbor of $u$.  Suppose $\sigma_u$ is applied on $u$. Then, $v$ is a $\sigma_u(t)$-neighbor of $u$. By our convention, we have $u$ is a $\overline{\sigma_u(t)}$-neighbor of $v$. Now we apply $\sigma_v$ on $v$, we get, $u$ is a $\sigma_v(\overline{\sigma_u(t)})$-neighbor of $v$, this implies, $v$ is a $\overline{\sigma_v(\overline{\sigma_u(t)})}$ neighbor of $u$. 
	
	\medskip
	
	Since $v$ is a $t$-neighbor of $u$, we have, $u$ is a $\overline{t}$-neighbor of $v$. Now we apply, $\sigma_v$ first on $v$, therefore, $u$ is a $\sigma_v(\overline{t})$-neighbor of $v$, and this is same as $v$ is a $\overline{\sigma_v(\overline{t})}$-neighbor of $u$, then we apply $\sigma_u$ on $u$, we have, $v$ is a $\sigma_u (\overline{\sigma_v(\overline{t})})$-neighbor of $u$. 
	
	If $\Gamma$ is switch-commutative, we get,
$$\overline{\sigma_v(\overline{\sigma_u(t)})} = \sigma_u(\overline{\sigma_v(\overline{t})})$$
Moreover, if $$\overline{\sigma_v(\overline{\sigma_u(t)})} = \sigma_u(\overline{\sigma_v(\overline{t})})$$
for all $\sigma_u, \sigma_v \in \Gamma$ and for all $t \in A_{n,m}$, then $\Gamma$ is switch-commutative (by definition). 
\end{proof}

\begin{corollary}\label{cor adjacency}
    Let $G$ be an $(n,m)$-graph and let $\Gamma$ be a switch-commutative group. Suppose $v$ is a $t$-neighbor of $u$ in $G$, and $\sigma_u$, $\sigma_v \in \Gamma$ is applied on $u$, $v$ respectively.  Then after the switches,  
    $v$ becomes a $\overline{\sigma_v(\overline{\sigma_u(t)}}$-neighbor of $u$ or equivalently, $v$ is a $\sigma_u(\overline{\sigma_v(\overline{t})})$-neighbor of $u$.
\end{corollary}

\begin{theorem}
    The $\Gamma$-switched graph $\rho_{\Gamma}(G)$ is well defined for all $(n,m)$-graphs $G$ if and only if $\Gamma$ is a switch-commutative group. 
\end{theorem}

\begin{proof}
    If $\Gamma$ is a switch-commutative group, it is clear that $\rho_{\Gamma}(G)$ is well defined for all $(n,m)$-graphs $G$. Suppose that for every $(n,m)$-graph $G$, the $\Gamma$-switched graph $\rho_{\Gamma}(G)$ is well defined. We prove that $\Gamma$ is a switch-commutative group. Let $v$ be a $t$-neighbor of $u$ for some $u,v \in V(G)$. Then by definition, $v^{\sigma_j}$ is a $\overline{\sigma_j(\overline{\sigma_i(t)})}$-neighbor of  $u^{\sigma_i}$ in $\rho_{\Gamma}(G)$ or $v^{\sigma_j}$ is a $\sigma_i(\overline{\sigma_j(\overline{t})})$-neighbor of  $u^{\sigma_i}$ in $\rho_{\Gamma}(G)$, where $v^{\sigma_j}$, $u^{\sigma_i}$ are vertices from the copies $G^{\sigma_j}$ and $G^{\sigma_i}$, for any $\sigma_i,\sigma_j \in \Gamma$ respectively. As, $\rho_{\Gamma}(G)$ is well defined, the order in which we switch the vertices should not matter, which forces $\overline{\sigma_j(\overline{\sigma_i(t)})} = \sigma_i(\overline{\sigma_j(\overline{t})})$, Thus, by Lemma~\ref{lem convention}, $\Gamma$ is a switch-commutative group. 
    \end{proof}

This $\Gamma$-switched graph helps build a bridge between 
$\langle e \rangle$-homomorphism and $\Gamma$-homomorphism of two 
$(n,m)$-graphs.  We prove a useful property of a switch-commutative group.

\begin{proposition}\label{prop chromatic num}
	Let $G$ and $H$ be two $(n,m)$-graphs. We have $G \xrightarrow{\Gamma} H$ if and only if $G \xrightarrow{\langle e \rangle} \rho_{\Gamma}(H)$, where $\Gamma$ is a switch-commutative group. 
\end{proposition}

\begin{proof}
For the ``only if'' part of the proof, suppose
	$f\colon G \xrightarrow{\Gamma} H$. Thus, $f\colon G' \xrightarrow{\langle e \rangle} H$ for some $G' \equiv_{\Gamma} G$. 
Consider, $\phi\colon G \rightarrow \rho_{\Gamma}(H)$ such that, 
	$$\phi(g) = f(g)^{{\sigma}^{-1}}$$
    where $\sigma$ is the switch applied on $g$ while obtaining $G'$ from $G$. 
We want to prove that $\phi$ is an $\langle e \rangle$-homomorphism. 

 Let $g_j$ be a $t$-neighbor of $g_i$ in $G$. 
 Suppose the switches $\sigma_i, \sigma_j \in \Gamma$ is applied on $g_i,g_j$ respectively to get $G'$. Thus by Corollary \ref{cor adjacency}, $$g_j  ~\text{is a}~ \overline{\sigma_j(\overline{\sigma_i(t)})}\text{-neighbor of}~ g_i ~\text{in}~ G'.$$ As $f$ is an $\langle e \rangle$-homomorphism, we have, 
 $$f(g_j) ~\text{is a}~\overline{\sigma_j(\overline{\sigma_i(t)})}\text{-neighbor of}~ f(g_i) ~\text{in}~ H.$$ 
 
 By the definition of $\rho_{\Gamma}(G)$, 
 $$f(g_j)^{{\sigma_j}^{-1}}~\text{is a}~ \overline{\sigma_j^{-1}(\overline{\sigma_i^{-1}(\overline{\sigma_j(\overline{\sigma_i(t)})})})}\text{-neighbor of}~ f(g_i)^{{\sigma_i}^{-1}} ~\text{in}~ \rho_{\Gamma}(H).$$ 
Thus, it is enough to prove, 
$$\overline{\sigma_j^{-1}(\overline{\sigma_i^{-1}(\overline{\sigma_j(\overline{\sigma_i(t)})})})} = t $$

Due to Lemma \ref{lem convention}, we have, 
$$\overline{\sigma_j(\overline{\sigma_i(t)})} = \sigma_i(\overline{\sigma_j(\overline{t})})$$
which implies 
$$\overline{\sigma_j^{-1}(\overline{\sigma_i^{-1}(\overline{\sigma_j(\overline{\sigma_i(t)})})})} =\overline{\sigma_j^{-1}(\overline{\sigma_i^{-1}(\sigma_i(\overline{\sigma_j(\overline{t})}))})} \\
= \overline{\sigma_j^{-1}(\sigma_j(\overline{t}))} \\
= \overline{\overline{t}} 
= t.$$

    This implies, $\phi$ is an $\langle e \rangle$-homomorphism an completes the ``only if'' part of the proof.

%    For the ``only if'' part of the proof, suppose
%    $f : G \xrightarrow{\Gamma} H$. Notice that the inclusion  
%    $i: H \xrightarrow{\Gamma} \rho_{\Gamma}(H)$ is a $\Gamma$-homomorphism. 
%    Thus, the composition function  $i \circ f$ is a $\Gamma$-homomorphism of $G$ to $\rho_{\Gamma}(H)$. 

    \medskip
    
    For the ``if''  part of the proof, 
   let  $f\colon G \xrightarrow{\langle e \rangle} \rho_{\Gamma}(H)$
   be an $\langle e \rangle$-homomorphism. 
   If $f(g) = h^{\sigma}$ for some $h \in V(H)$ and some $\sigma \in \Gamma$, then set $\varphi(g) = h$, for all $g \in V(G)$. 
   Moreover, construct $G'$ from $G$ by applying $\sigma$ on $g$, for all $g \in V(G)$. 
Since $V(G) = V(G')$ and $G'$ is $\Gamma$-equivalent to $G$,  it is enough to show   that  $\varphi$ is an $\langle e \rangle$-homomorphism of $G'$ to $H$.

  Suppose $f(g_i) = h_i^{\sigma_i}$ and $f(g_j)=h_j^{\sigma_j}$ and let  $g_j$ be a $t$-neighbor of $g_i$ in $G'$.    
    We now have to show that 
   $h_j$  is a $t$-neighbor of $h_i$ in $H$. 
   As
$$g_j ~\text{is a}~ t\text{-neighbor of}~  g_i~\text{in}~ G',$$
by Corollary \ref{cor adjacency}, we have, 
$$g_j ~\text{is a}~ 
\overline{\sigma_j^{-1}(\overline{\sigma_i^{-1}(t)}})\text{-neighbor of}~ g_i ~\text{in}~ G.$$
As $f$ is an $\langle e \rangle$-homomorphism, 
$$f(g_j) ~\text{is a}~ 
\overline{\sigma_j^{-1}(\overline{\sigma_i^{-1}(t)}})\text{-neighbor of}~ f(g_i) ~\text{in}~ \rho_{\Gamma}(H).$$ That is, 
$$ h_j^{\sigma_j}~\text{is a}~ \overline{\sigma_j^{-1}(\overline{\sigma_i^{-1}(t)}})\text{-neighbor of}~ h_i^{\sigma_i} ~\text{in}~ \rho_{\Gamma}(H).$$
Thus, by Corollary \ref{cor adjacency}, we have, 
$$ h_j~\text{is a}~ 
\overline{\sigma_j(\overline{\sigma_i(\overline{\sigma_j^{-1}(\overline{\sigma_i^{-1}(t)}}))})}\text{-neighbor of}~ h_i ~\text{in}~ H.$$
Therefore, it is enough to prove, 
$$\overline{\sigma_j(\overline{\sigma_i(\overline{\sigma_j^{-1}(\overline{\sigma_i^{-1}(t)}}))})} =t$$
By Lemma \ref{lem convention}, we have, 
$$\overline{\sigma_j(\overline{\sigma_i(t)})} = \sigma_i(\overline{\sigma_j(\overline{t})})$$
which implies 
$$\sigma_i(\overline{\sigma_j(\overline{\overline{\sigma_j^{-1}(\overline{\sigma_i^{-1}(t)}}}))})
= \sigma_i(\overline{\sigma_j(\sigma_j^{-1}(\overline{\sigma_i^{-1}(t)}))})
= \sigma_i(\overline{\overline{\sigma_i^{-1}(t)}})
= \sigma_i(\sigma_i^{-1}(t)) =t$$
Thus,  
$h_j$ is a $t$-neighbor of $h_i$ in $H$. This completes the ``if'' part of the proof. 
\end{proof}

\begin{theorem}\label{th isomorphism}

	Let $G$ and $H$ be $(n,m)$-graphs. Then, $G \equiv_{\Gamma} H$ if and only if  $\rho_{\Gamma}(G) \equiv_{\langle e \rangle} \rho_{\Gamma}(H)$,  where $\Gamma$ is a switch-commutative group.  

\end{theorem}

\begin{proof}
For the ``only if" part of the proof, suppose that $G \equiv_{\Gamma} H.$ 
Let $G'$ be a $\Gamma$-equivalent graph of $G$ such that 
$f\colon G' \xrightarrow{\langle e \rangle} H$ is 
an $\langle e \rangle$-isomorphism. 
Suppose that $G$ has vertices $g_1, g_2, \ldots, g_p$, and $G'$ is obtained by performing a 
$\tau_i$-switch on $g_i$, where $g_i \in V(G)$ and $\tau_i \in \Gamma$. For convenience,
the particular switch applied on $g_i$ is called $\tau_i$ 

For any $g_{i}^{\sigma_{i}} \in V(\rho_{\Gamma}(G))$
 we define the function $\phi\colon V(\rho_{\Gamma}(G)) \rightarrow V(\rho_{\Gamma}(H))$ as follows: 
 $$\phi(g_{i}^{\sigma_i}) = \left(f(g_{i})\right)^{\tau_i^{-1}\sigma_i}.$$ 
 
 Next we prove that $\phi$ is a $\langle e \rangle$-isomorphism of 
 $\rho_{\Gamma}(G)$ to $\rho_{\Gamma}(H)$.  

Let $g_j^{\sigma_j}$ be a $t$-neighbor of $g_i^{\sigma_i}$ in 
$\rho_{\Gamma}(G)$. Then by Corollary \ref{cor adjacency},
$$g_j ~\text{is a}~ \overline{\sigma_j^{-1}(\overline{\sigma_i^{-1}(t)})}\text{-neighbor of}~ g_i ~\text{in}~ G$$
and
$$g_j ~\text{is a}~ \overline{\tau_j(\overline{\tau_i(\overline{\sigma_j^{-1}(\overline{\sigma_i^{-1}(t)})})})} ~\text{-neighbor of}~ g_i ~\text{in}~ G'.$$ 
As $f$ is a $\langle e \rangle$-isomorphism, 
$$f(g_j) ~\text{is a}~ \overline{\tau_j(\overline{\tau_i(\overline{\sigma_j^{-1}(\overline{\sigma_i^{-1}(t)})})})} ~\text{-neighbor of}~ f(g_i) ~\text{in}~ H.$$ 
Now, 
$$(f(g_j))^{\sigma_j} ~\text{is a}~ \overline{\sigma_j(\overline{\sigma_i(\overline{\tau_j(\overline{\tau_i(\overline{\sigma_j^{-1}(\overline{\sigma_i^{-1}(t)})})})})})} ~\text{-neighbor of}~ (f(g_i))^{\sigma_i} ~\text{in}~ \rho_{\Gamma}(H).$$
Thus, $$\phi(g_{j}^{\sigma_j}) = (f(g_j))^{\tau_j^{-1}\sigma_j}  ~\text{is a}~ \overline{\tau_j^{-1}\sigma_j(\overline{\tau_i^{-1}\sigma_i(\overline{\tau_j(\overline{\tau_i(\overline{\sigma_j^{-1}(\overline{\sigma_i^{-1}(t)})})})})})} ~\text{-neighbor of}~ \phi(g_{i}^{\sigma_i}) = (f(g_i))^{\sigma_i} .$$ 
As $\Gamma$ is switch-commutative, we get,
$$\overline{\tau_j^{-1}\sigma_j(\overline{\tau_i^{-1}\sigma_i(\overline{\tau_j(\overline{\tau_i(\overline{\sigma_j^{-1}(\overline{\sigma_i^{-1}(t)})})})})})} = \overline{\sigma_j(\overline{\underbracket{\sigma_i(\overline{\tau_j^{-1}(\overline{\underbracket{\tau_i^{-1}(\overline{\tau_j(\overline{\underbracket{\tau_i(\overline{\sigma_j^{-1}(\overline{\sigma_i^{-1}(t)})})}})})}})})}})}.$$
By repeated application of Lemma~\ref{lem convention} on $(1),(2)$ and $(3)$, we get, 
$$   \overline{\sigma_j(\overline{\underbracket{\sigma_i(\overline{\tau_j^{-1}(\overline{\underbracket{\tau_i^{-1}(\overline{\tau_j(\overline{\underbracket{\tau_i(\overline{\sigma_j^{-1}(\overline{\sigma_i^{-1}(t)})})}_{(1)}})})}_{(2)}})})}_{(3)}})} = t$$
 Therefore,  $\phi(g_j^{\sigma_j})$ is a $t$-neighbor of 
 $\phi(g_i^{\sigma_i})$ in $\rho_{\Gamma(H)}$. 
Thus, $\phi$ a $\langle e \rangle$-isomorphism of $\rho_{\Gamma}(G)$ to $\rho_{\Gamma}(H)$.

%Note that  $g_{i}^{\sigma_j}$ is a $t$-neighbor of  $g_k^{\sigma_l}$ in $\rho_{\Gamma}(G)$ if and only if  
%$g_{i}$ is a $\sigma_{j}^{-1}(\sigma_{l}^{-1}(t))$-neighbor of $g_k$ in $G$. This is if and only if to the following:  
%$g_{i}$ is a $\tau_i(\tau_k(\sigma_{j}^{-1}(\sigma_{l}^{-1}(t))))$-neighbor of $g_k$ in $G'$. 
%Since $f$ is an $\langle e \rangle$-isomorphism of $G'$ to $H$, the previous statement is if and only if  $f(g_{i})$ is a 
%$\tau_i(\tau_k(\sigma_{j}^{-1}(\sigma_{l}^{-1}(t))))$-neighbor of $f(g_k)$ in $H$. 
%This is true if and only if $\phi(g_{i}^{\sigma_j}) = \left(f(g_{i})^{\tau_i^{-1}}\right)^{\sigma_j}$ is a $\gamma(t)$-neighbor of $\phi(g_{k}^{\sigma_l}) = \left(f(g_{k})^{\tau_k^{-1}}\right)^{\sigma_l}$  in $\rho_{\Gamma}(H)$ where 
%$$\gamma(t) = \sigma_l(\tau_k^{-1}(\sigma_j(\tau_i^{-1}(\tau_i(\tau_k(\sigma_{j}^{-1}(\sigma_{l}^{-1}(t)))))))).$$ 
%However, as $\Gamma$ is Abelian, $\gamma(t)=t$. 
%Thus, $\phi$ a $\langle e \rangle$-isomorphism of 
%$\rho_{\Gamma}(G)$ to $\rho_{\Gamma}(H)$. 

\bigskip

For the ``if'' part of the proof, suppose $\rho_{\Gamma}(G) \equiv_{\langle e \rangle} \rho_{\Gamma}(H)$ and we have to show $G \equiv_{\Gamma} H$. 

Assume $g_{1}, g_{2}, \ldots, g_{p}$ be the vertices of $G$. A sequence of vertices in $\rho_{\Gamma}(G)$
of the form $(g_1^{\sigma_1}, g_2^{\sigma_2}, \ldots, g_p^{\sigma_p})$ is a representative sequence of $G$ in $\rho_{\Gamma}(G)$, where $\sigma_i \in \Gamma$ is any element for $i \in \{1,2, \ldots, p\}$ (repetition of elements among $\sigma_i$ is allowed here).

Given an $\langle e \rangle$-isomorphism  
$\psi\colon \rho_{\Gamma}(G) \xrightarrow{\langle e \rangle} \rho_{\Gamma}(H)$ and a 
representative sequence $S$ of $G$ in $\rho_{\Gamma}(G)$, 
define the set 
$$Y_{S, \psi} = \{ v^{\sigma} ~|~ \psi(v^{\sigma}) = (\psi(v))^{\sigma} \text{ where }  v \in S \text{ and } \sigma \in \Gamma \}.$$
Let $Y_{S^*, \varphi}$ be the set satisfying the property 
$|Y_{S^*, \varphi}| \geq |Y_{S, \psi}|$ 
where $S$ varies over all representative sequences and $\psi$ varies over all $\langle e \rangle$-isomorphisms.

We proceed by the method of contradiction to show that $Y_{S^*, \varphi} = V(\rho_{\Gamma}(G))$. Thus, let us assume the contrary, that is, let 
$Y_{S^*, \varphi} \neq V(\rho_{\Gamma}(G))$. 
This implies that there exists a $v^\sigma \in V(\rho_{\Gamma}(G))$, for some $v \in S$ and some $\sigma \in \Gamma$
such that $\varphi(v^{\sigma}) \neq (\varphi(v))^{\sigma}$. Let us fix
$g = \varphi^{-1}(\varphi(v)^{\sigma}) \in \rho_{\Gamma}(G)$. 
Next let us define the function

	\[\widehat{\varphi}(x) = \begin{cases}
 		\varphi(x) &\text{ if }x \neq g, v^{\sigma}, \\
 		\varphi(g) &\text{ if } x = v^{\sigma}, \\
 		\varphi(v^{\sigma}) &\text{ if } x = g.
 	\end{cases}\]

Next we are going to show that $\widehat{\varphi}$ is an $\langle e \rangle$-isomorphism of 
$\rho_{\Gamma}(G)$ and $\rho_{\Gamma}(H)$. So, we need to show that 
 $x$ is a $t$-neighbor of $y$ in $\rho_{\Gamma}(G)$
 if and only if 
 $\widehat{\varphi}(x)$ is a $t$-neighbor of $\widehat{\varphi}(y)$ in $\rho_{\Gamma}(H)$.
 Notice that, it is enough to check this for $x=g$ and $x=v^{\sigma}$ while $y$ varies 
 over all vertices of $\rho_{\Gamma}(G)$. We will separately handle the exceptional case when $x=v^{\sigma}$ and $y = g$ first. 
 
 \begin{enumerate}[(i)]
     \item  If $x = v^{\sigma}$ and $y = g$, then
     $\varphi(v)$ and $\varphi(v)^{\sigma}$ are non-adjacent by the definition of $\rho_{\Gamma}(G)$. Thus, 
     $v = \varphi^{-1}(\varphi(v))$ and
     $g = \varphi^{-1}(\varphi(v)^{\sigma})$ are non-adjacent. Hence  $x=v^{\sigma}$ and $y=g$ are also non-adjacent. Therefore,  
     $\widehat{\varphi}(x) = \varphi(g)$
     and $\widehat{\varphi}(y) = \varphi(v^{\sigma})$ are non-adjacent. 
     
     \item If $x = v^{\sigma}$ and $y \neq g$, then $y$ is a $t$-neighbour of $x$ in $\rho_{\Gamma}(G)$ if and only if $\varphi(y)$ is a $t$-neighbor of $\varphi(x)$ in $\rho_{\Gamma}(H)$, as $\varphi$ is an $\langle e \rangle$-isomorphism. 
     Observe that  $\widehat{\varphi}(x) = \varphi(g) = \varphi(v)^{\sigma}$ 
     as $g = \varphi^{-1}(\varphi(v)^{\sigma})$, 
     and $\widehat{\varphi}(y) = \varphi(y)$.
      Since $x = v^{\sigma}$, $y$ is a $\sigma^{-1}(t)$-neighbor of $v$ in $\rho_{\Gamma}(G)$, if and only if $\varphi(y)$ is a $\sigma^{-1}(t)$-neighbor of $\varphi(v)$ in $\rho_{\Gamma}(H)$, if and only if, $\varphi(y) = \widehat\varphi(y)$ is a $t$-neighbour of $\varphi(v)^{\sigma} = \widehat{\varphi}(x)$ in $\rho_{\Gamma}(H)$.
     
     \item  If $x = g$ and $y \neq v^{\sigma}$, then $y$ is a $t$-neighbor of $x$ in $\rho_{\Gamma}(G)$ if and only if, $\varphi(y)$ is a $t$-neighbor of $\varphi(x)= \varphi(g) = \varphi(v)^{\sigma}$ in $\rho_{\Gamma}(H)$, as $g = \varphi^{-1}(\varphi(v)^{\sigma})$. 
     The previous statement holds if and only if 
     $\varphi(y)$ is a $\sigma^{-1}(t)$-neighbor of $\varphi(v)$ in $\rho_{\Gamma}(H)$ if and only if $y$ is a $\sigma^{-1}(t)$-neighbor of $v$ in $\rho_{\Gamma}(G)$ if and only if  $y$ is a $t$-neighbor of $v^{\sigma}$ in $\rho_{\Gamma}(G)$ if and only if $\varphi(y)=\widehat\varphi(y)$ is a $t$-neighbor of $\varphi(v^{\sigma}) = \widehat\varphi(x)$ in $\rho_{\Gamma}(H)$ . 
 \end{enumerate}

However, now we have $|Y_{S^*, \varphi}| < |Y_{S^*, \widehat\varphi}|$. This is a contradiction to the definition of $Y_{S^*, \varphi}$, and hence 
$Y_{S^*, \varphi} = V(\rho_{\Gamma}(G))$.

\medskip

Let $v_1, v_2$ be two distinct vertices in $S^*$. If 
$\varphi(v_1)^\sigma = \varphi(v_2)$ for any $\sigma \in \Gamma$, then 
$\varphi(v_1^\sigma) = \varphi(v_2)$. This implies 
$v_1^\sigma = v_2$ because $\varphi$ is a bijection. However, this is not possible as $v_1, v_2$ are distinct vertices from a representative sequence $S^*$ of $G$ in $\rho_{\Gamma}(G)$. Hence, 
$\varphi(v_1)^\sigma \neq \varphi(v_2)$ for any $v_1, v_2 \in S^*$. That means, $\varphi(S^*) = R$ is a representative sequence of $H$ in $\rho_{\Gamma}(H)$. 
Thus, note that $\langle e \rangle$-isomorphism restricted to the induced subgraph $\rho_{\Gamma}(G)[S^*]$ is also an $\langle e \rangle$-isomorphism to the induced subgraph $\rho_{\Gamma}(H)[R]$. 
That is, $\rho_{\Gamma}(G)[S^*] \equiv_{\langle e \rangle}  \rho_{\Gamma}(H)[R]$.
As $\langle e \rangle \subseteq \Gamma$, this also means  
$\rho_{\Gamma}(G)[S^*] \equiv_{\Gamma}  \rho_{\Gamma}(H)[R]$.

Moreover, as $S^*$ and $R$ are representative sequences of $G$ and $H$, respectively, we have $\rho_{\Gamma}(G)[S^*] \equiv_{\Gamma} G$ and  $\rho_{\Gamma}(H)[R] \equiv_{\Gamma} H$. Thus we are done by composing the $\Gamma$-isomorphisms. 
\end{proof}

The above result generalizes results by Brewster and Graves~\cite{brewster2009edge} (see Theorem~$12$) and 
Sen~\cite{sen2017homomorphisms} (see Theorem~$3.4$). Additionally, it (re)solves an open problem given by  Klostermeyer and MacGillivray~\cite{klostermeyer2004homomorphisms} (see Open Problem~$2$ in the conclusion) by restricting the result to 
$(n,m)=(1,0)$, where $\Gamma$ is the group in which the only non-identity element simply  reverses the direction of the arcs.

The next result follows from the fundamental theorem of finite abelian groups.

\begin{theorem}
   Let $\Gamma_1$ be an abelian, 
   $(n,m)$-switch-commutative group. 
   Let $\Gamma_2 \subseteq \Gamma_1$. If $p^{2} \nmid |\Gamma_1|$ for any prime $p$, then $\rho_{\Gamma_1}(G) \equiv_{\langle e \rangle} \rho_{\Gamma_1/\Gamma_2}(\rho_{\Gamma_2}(G))$.
\end{theorem}

\begin{proof}
 Since $\Gamma_1$ is a finite Abelian group, $\Gamma_1/\Gamma_2$ and $\Gamma_2$ both are normal subgroups of $\Gamma_1$,  As, $p^{2} \nmid |\Gamma_1|$, we have $ \Gamma_1/\Gamma_2 \times \Gamma_2 \cong \Gamma_1$. Observe that 
 every element $ \sigma \in \Gamma_1$ can be uniquely written as  $\alpha . \beta$, where $\alpha \in \Gamma_1/\Gamma_2, \beta \in \Gamma_2$ (unique factorization). 
 Now let $G$ be an $(n,m)$-graph. We prove, $f\colon \rho_{\Gamma_1}(G) \to \rho_{\Gamma_1/\Gamma_2}(\rho_{\Gamma_2}(G))$ is an isomorphism. Consider, 
 $$f\colon V(\rho_{\Gamma_1}(G)) \to V(\rho_{\Gamma_1/\Gamma_2}(\rho_{\Gamma_2}(G)),$$
 $$f(u^{\sigma}) = (u^{\beta})^{\alpha}.$$ where $\alpha \in \Gamma_1/\Gamma_2, \beta \in \Gamma_2$, and $\alpha . \beta = \sigma$.

Let $\sigma_1 = \alpha_1.\beta_1$ and $\sigma_2=\alpha_2.\beta_2$, where $\sigma_1, \sigma_2 \in \Gamma_1$, $\alpha_1, \alpha_2 \in \Gamma_1/\Gamma_2$, $\beta_1, \beta_2 \in \Gamma_2$. 
Now consider the following two sequences of switches on $u$ and $v$, respectively. 
The first sequence is: 
$\beta_1$ applied on $u$, 
$\alpha_1$ applied on $u$, 
$\beta_2$ applied on $v$, 
$\alpha_2$ applied on $v$. 
Suppose as a result of the above-mentioned sequence of switches, $v$ becomes an $s$-neighbor of $u$. 
The second sequence is: 
$\beta_1$ applied on $u$, 
$\beta_2$ applied on $v$, 
$\alpha_1$ applied on $u$, 
$\alpha_2$ applied on $v$. 
Since $\Gamma_1$ is switch-commutative, 
as a result of the above-mentioned sequence of switches, $v$ must become an 
$s$-neighbor 
(from initially being a $t$-neighbor)
of $u$. 
From the above observation we can conclude that
$v^{\sigma_2}$ is an $s$-neighbor of 
$u^{\sigma_1}$ in $\rho_{\Gamma_1}(G)$
and 
$(v^{\beta_2})^{\alpha_2}$ 
is an 
$s$-neighbor of 
$(u^{\beta_1})^{\alpha_1}$ in 
$\rho_{\Gamma_1/\Gamma_2}(\rho_{\Gamma_2}(G))$.
This proves that $f$ is an isomorphism. 
\end{proof}

% From the above theorem and from the fundamental theorem of group homomorphism, we get the following corollary.  

% \begin{corollary}
% 	Let $\Gamma_{1} \subseteq S_{2n+m}$ and let $\Gamma_2$ be a normal subgroup of $\Gamma_1$. 
% Then  
% $$\rho_{\Gamma_1}(G) \equiv \rho_{\Gamma_1/\Gamma_2}(\rho_{\Gamma_2}(G)).$$ 
% \end{corollary}

% \begin{proof}
 
% \end{proof}

A \textit{$\Gamma$-core} of an $(n,m)$-graph $G$ is a subgraph $H$ of $G$ such that  
$G \xrightarrow{\Gamma} H$,  whereas $H$ does not admit a $\Gamma$-homomorphism to any of its proper subgraphs.

\begin{theorem}
The core of an $(n,m)$-graph $G$ is unique up to $\Gamma$-isomorphism. 
\end{theorem}

\begin{proof}
Let $H_1$ and $H_2$ be two $\Gamma$-cores of $G$. We have to show that $H_1$ and $H_2$ are $\Gamma$-isomorphic. 

Note that, there exist  $\Gamma$-homomorphisms 
$f_1\colon G \xrightarrow{\Gamma} H_1$
and 
$f_2\colon G \xrightarrow{\Gamma} H_2$ 
as $H_1, H_2$ are $\Gamma$-cores. 
Moreover, there exists the inclusion $\Gamma$-homomorphisms
$i_1\colon H_1 \xrightarrow{\Gamma} G$
and 
$i_2\colon H_2 \xrightarrow{\Gamma} G$. 

Now consider the composition $\Gamma$-homomorphism 
$f_2 \circ i_1\colon H_1 \xrightarrow{\Gamma} H_2$. Note that it must be a surjective vertex mapping. 
Not only that, for any non-adjacent pair $u,v$ of vertices in $H_1$, the vertices $(f_2 \circ i_1)(u)$ and $(f_2 \circ i_1)(v)$ are non-adjacent in $H_2$. The reason is that, if the above two conditions are not satisfied, then the composition $\Gamma$-homomorphism $f_2 \circ i_1 \circ f_1\colon G \xrightarrow{\Gamma} H_2$ can be considered as a $\Gamma$-homomorphism to a proper subgraph of $H_2$. This will contradict the fact that $H_2$ is a $\Gamma$-core. Therefore, $f_2 \circ   i_1$ is a bijective $\Gamma$-homomorphism whose inverse is also a $\Gamma$-homomorphism. In other words, $f_2 \circ i_1$ is a $\Gamma$-isomorphism. 
\end{proof}

Due to the above theorem, it is possible to define the $\Gamma$-core of $G$ and let us denote it by $core_{\Gamma}(G)$. Notice that, this is the analogue of the fundamental algebraic concept of core in the study of graph homomorphism. 

\section{Categorical products}\label{sec category}
Taking the set of $(n,m)$-graphs as objects and 
their $\Gamma$-homomorphisms as morphisms, one can consider the category of $(n,m)$-graphs with respect to $\Gamma$-homomorphism. In this section, we study whether products and co-products exist in this category or not. 
The existence of categorical product and co-product will not only contribute in establishing the category of $(n,m)$-graphs with respect to $\Gamma$-homomorphism as a richly structured category, but it will also show that
the lattice of $(n,m)$-graphs induced by $\Gamma$-homomorphisms is a distributive lattice with the categorical products and co-products playing the roles of join and meet, respectively. Moreover,
categorical product~\cite{hell2004graphs} was useful in proving the density theorem~\cite{nevsetvril2002density} for undirected and directed graphs. Thus, it is not wrong to hope that it 
may become useful to prove the analogue of the density theorem in our context. It is worth commenting that the the idea to prove the existence of categorical products in this context generalizes the idea of the same in the context of signed graphs from~\cite{senorder}.

Before proceeding further with the results, let us recall what categorical product and co-product mean in our context. 
Let $G, H$ be two $(n,m)$-graphs and let $\Gamma \subseteq S_{2n+m}$ be an Abelian group.

The \textit{categorical product} of $G$ and $H$ with respect to 
$\Gamma$-homomorphism is an $(n,m)$-graph $P$ having two projection mappings of the form 
$f_g\colon P \xrightarrow{\Gamma} G$ and $f_h\colon P \xrightarrow{\Gamma} H$ satisfying the following universal property: if any $(n,m)$-graph $P'$
admit $\Gamma$-homomorphisms 
$\phi_g\colon P' \xrightarrow{\Gamma} G$ 
and 
$\phi_h\colon P' \xrightarrow{\Gamma} H$, then there exists a unique 
$\Gamma$-homomorphism $\varphi\colon P' \xrightarrow{\Gamma} P$ such that 
$\phi_g = f_g \circ \varphi$ and $\phi_h = f_h \circ \varphi$. Refer Figure~\ref{Fig: product} for its commutative diagram. 

The \textit{categorical co-product} of $G$ and $H$ with respect to 
$\Gamma$-homomorphism is an $(n,m)$-graph $C$ along with the two inclusion mappings  of the form 
$i_g\colon G \xrightarrow{\Gamma} C$ and $i_h\colon H \xrightarrow{\Gamma} C$ satisfying the following universal property: 
if for any $(n,m)$-graph $C'$
there are $\Gamma$-homomorphisms  
$\phi_g\colon G \xrightarrow{\Gamma} C'$ 
and 
$\phi_h\colon H \xrightarrow{\Gamma} C'$, then there exists a unique 
$\Gamma$-homomorphism $\varphi\colon C \xrightarrow{\Gamma} C'$ such that 
$\phi_g = \varphi \circ  i_g$ and $\phi_h =  \varphi  \circ i_h$. Refer Figure~\ref{Fig: coproduct} for its commutative diagram.
\begin{figure}
\begin{center}
    \tikzset{column sep=small, ampersand replacement=\&}
\begin{floatrow}
    \centering
    \ffigbox{\begin{tikzcd}
\& G 
		\&
		\&[1.5em] \\
		P \ar[ur, "f_g"] \ar[dr, "f_h"'] 
		\&
		\& P' \ar[ul, "\phi_g" '] \ar[dl, "\phi_h"] \ar[ll, "\varphi" ']
		\&  \\
		\& H 
		\&
		\&
	\end{tikzcd}}{\caption{Product of $G$ and $H.$} \label{Fig: product}}
     \ffigbox{\begin{tikzcd}
     \& G \ar[dl, "i_g" '] \ar[dr, "\phi_g"]
		\&
		\&[1.5em] \\
		C  \ar[rr, "\varphi"]
		\&
		\& C'  
		\&  \\
		\& H \ar[ul, "i_h"] \ar[ur, "\phi_h" ']
		\&
		\&
	\end{tikzcd}}{\caption{Co-product of $G$ and $H.$}\label{Fig: coproduct}}
\end{floatrow}
\end{center}
\end{figure}

Let $G, H$ be two $(n,m)$-graphs and let $\Gamma \subseteq S_{2n+m}$ be a switch-commutative group. Then  
$G \times_{\langle e \rangle} H$ denotes 
the $(n,m)$-graph on the vertex set  $V(G) \times V(H)$, where 
$(u,v)$ is a $t$-neighbor of $(u',v')$ in $G \times_{\langle e \rangle} H$
if and only if $u$ is a $t$-neighbor of $u'$ in $G$ and $v$ is a $t$-neighbor of $v'$ in $H$. 
Moreover, the $(n,m)$-graph 
$G \times_{\Gamma} H$ is the subgraph of 
$\rho_{\Gamma}(G) \times_{\langle e \rangle} \rho_{\Gamma}(H)$ induced by the  set of vertices 
$$X = \{(u^{\sigma}, v^{\sigma})\colon (u,v) \in V(G) \times V(H) \text{ and }  \sigma \in \Gamma\}.$$ 

\begin{theorem}\label{th product existence}
The categorical product of $(n,m)$-graphs $G$ and $H$ with respect to $\Gamma$-homomorphism exists and is $\Gamma$-isomorphic to 
$G \times_{\Gamma} H$, where $\Gamma$ is a switch-commutative group.
\end{theorem}

\begin{proof}
Let $(G \times_{\Gamma} H)'$ be $\Gamma$-switched graph of $G \times_{\Gamma} H$, where we apply $\sigma^{-1}$ on $(u^{\sigma},v^{\sigma}) \in V(G \times_{\Gamma} H)$. Thus, we define $f_g(u^{\sigma},v^{\sigma}) = u$ and $f_h(u^{\sigma},v^{\sigma})= v$ as the two projections. Observe that $f_g$ and $f_h$ are $\langle e \rangle$-homomorphisms of 
$(G \times_{\Gamma} H)'$ to $G$ and $H$, respectively. If there exists an $(n,m)$-graph $P'$ such that, $\phi_g\colon P' \xrightarrow{\Gamma} G$ and $\phi_h\colon P' \xrightarrow{\Gamma} H$, then define $\phi\colon P' \xrightarrow{\Gamma} G \times_{\Gamma} H$ such that $\phi(p) = (\phi_g(p), \phi_h(p))$. From the definition of $\phi$, we have $\phi_g = f_g \circ \phi$ and $\phi_h = f_h \circ \phi$. Note that this is the unique way we can define $\phi$ which satisfies the universal property from the definition of products. Thus, $G \times_{\Gamma} H$ is indeed the categorical product of $G$ and $H$ with respect to $\Gamma$-homomorphism once we prove its uniqueness up to 
$\Gamma$-isomorphism. 

Suppose $P_1$ with projection mappings $f_g, f_h$ and $P_2$ with projection mappings $\phi_g, \phi_h$ are two $(n,m)$-graphs that satisfy the universal properties of categorical product of $G$ and $H$, then there exists $\varphi\colon P_1 \xrightarrow{\Gamma} P_2$ and $\varphi'\colon P_2 \xrightarrow{\Gamma} P_1$ with $\phi_g \circ \varphi = f_g$, $\phi_h \circ \varphi = f_h$ and $f_g \circ \varphi' = \phi_g$, $f_h \circ \varphi' = \phi_h$. Now consider the composition, $\varphi' \circ \varphi\colon P_1 \xrightarrow{\Gamma} P_1$. Since, $f_g \circ (\varphi' \circ \varphi) = (f_g \circ \varphi') \circ \varphi 
= \phi_g \circ \varphi 
= f_g ~\text{and}~ f_h \circ (\varphi' \circ \varphi) = (f_h \circ \varphi') \circ \varphi = \phi_h \circ \varphi' 
= f_h$

We can conclude that $\varphi' \circ \varphi$ is a  identity homomorphism on $P_1$. Similarly $\varphi \circ \varphi'$ must be the identity homomorphism on $P_2$. Thus implying, $\varphi' = \varphi^{-1}$ is an $\Gamma$-isomorphism of $P_2$ and $P_1$. 
\end{proof}

In PEPS 2012 workshop, Brewster asked whether 
Categorical product exists for signed graphs or not. The above theorem answers this question in affirmative as a special case (with respect to $\Gamma$-homomorphisms of $(0,2)$-graphs, where $\Gamma$ is non-trivial).

\begin{corollary}
For any $(n,m)$-graphs $G$ and $H$, and a switch-commutative group $\Gamma \subseteq S_{2n+m}$, we have  $\rho_{\Gamma}(G \times_{\Gamma} H) \equiv_{\langle e \rangle} \rho_{\Gamma}(G) \times_{\langle e \rangle} \rho_{\Gamma}(H)$. 
\end{corollary}

\begin{proof}
    The proof follows from Theorem~\ref{th isomorphism} and Theorem~\ref{th product existence}. 
\end{proof}

% \begin{figure}[ht!]
% \begin{center}
% \begin{tikzcd}[column sep=normal]
%  & P' \ar[dl, "\Gamma" '] \ar[d, "\Gamma"] \ar[dr, "\Gamma"] &  & & X \ar[dl, "\langle e \rangle" '] \ar[d, "\langle e \rangle"] \ar[dr, "\langle e \rangle"] & \\
%   G  & G \times_{\Gamma} H \ar[l, "\Gamma" '] \ar[r, "\Gamma"] &  H &

%   \rho_{\Gamma}(G)  & \rho_{\Gamma}(G) \times_{\langle e \rangle} \rho_{\Gamma}(H) \ar[l, "\langle e \rangle" '] \ar[r, "\langle e \rangle"] &  \rho_{\Gamma}(H) 
% \end{tikzcd}
% \end{center}
% \caption{Pictorial representation of the commutative maps.}
% \label{fig comm diagram}

% \end{figure}

%G   &  G \times_{\Gamma} H \ar[ll, {\Gamma}] \ar[rr, {\Gamma}] & H

Let $G+H$ denote the disjoint union of the $(n,m)$-graphs $G$ and $H$. 

\begin{theorem}\label{th coproduct existence}
The categorical co-product of $(n,m)$-graphs $G$ and $H$ with respect to $\Gamma$-homomorphism exists and is $\Gamma$-isomorphic to 
$G + H$, where $\Gamma$ is any subgroup of $S_{2n+m}$.
\end{theorem}

\begin{proof}
	Consider the inclusion mapping $i_g\colon G \xrightarrow{\Gamma} G + H$, and $i_h\colon  H \xrightarrow{\Gamma} G + H$. Suppose there exists an $(n,m)$-graph $C$ and there are $\Gamma$-homomorphisms $\phi_g\colon G \xrightarrow{\Gamma} C$ 
and $\phi_h\colon H \xrightarrow{\Gamma} C$, then there exists a  $\Gamma$-homomorphism $\varphi\colon G + H \xrightarrow{\Gamma} C$ such that 
\[ \varphi(x) = \begin{cases}
\phi_g(x) &\text{ if } x \in V(G), \\
\phi_h(x) &\text{ if } x \in V(H).

\end{cases} \]

Observe that, $\phi_g = \varphi \circ  i_g$ and $\phi_h =  \varphi  \circ i_h$. Note that such a $\varphi$ is unique. 

Suppose we have $P$ with $\Gamma$-homomorphisms $f_g, f_h$ and $P'$ with $\Gamma$-homomorphisms $\phi_g, \phi_h$ satisfying the universal property of categorical co-product of $G$ and $H$, then there exists $\Gamma$-homomorphisms, $\varphi\colon P \xrightarrow{\Gamma} P'$ and $\varphi' : P' \xrightarrow{\Gamma} P$ with $\phi_g = \varphi \circ f_g$, $\phi_h = \varphi \circ f_h$ and $f_g = \varphi' \circ \phi_g$, $f_h = \varphi' \circ \phi_h$. Now consider the composition, $\varphi' \circ \varphi\colon P \xrightarrow{\Gamma} P$. Since, $(\varphi' \circ \varphi) \circ f_g = \varphi' \circ ( \varphi \circ f_g) = \varphi' \circ \phi_g = f_g$ and $(\varphi' \circ \varphi) \circ f_h =  \varphi' \circ (\varphi  \circ f_h) = \varphi' \circ \phi_h = f_h$
 We can say that $\varphi' \circ \varphi$ is the identity mapping on $P$. Similarly $\varphi \circ \varphi'$ must be the identity mapping on $P$. Thus implying, $\varphi' = \varphi^{-1}$ is an $\Gamma$-isomorphism of $P'$ and $P$. Therefore, we have, the categorical co-product of $(n,m)$-graphs $G$ and $H$ with respect to $\Gamma$-homomorphism is $G + H$. 
\end{proof}

Thus both categorical product and co-product exists with respect to $\Gamma$-homomorphism when $\Gamma$ is a switch-commutative group. Furthermore, the usual algebraic identities hold with respect to these operations too.

\begin{corollary}
For any $(n,m)$-graphs $G$ and $H$, and a switch-commutative group $\Gamma \subseteq S_{2n+m}$, we have 
$\rho_{\Gamma}(G + H) \equiv_{\langle e \rangle} \rho_{\Gamma}(G) + \rho_{\Gamma}(H)$. 
\end{corollary}

\begin{proof}
    The proof follows from
    Theorem~\ref{th isomorphism} and Theorem~\ref{th coproduct existence}.  
\end{proof}

\begin{theorem}
For any $(n,m)$-graphs $G, H, K$ and $\Gamma$, a switch-commutative group, We have the following. 
\begin{enumerate}[(i)]
    \item $G \times_{\Gamma} H \equiv_{\Gamma} H \times_{\Gamma} G$,
    
    \item $G \times_{\Gamma}(H \times_{\Gamma} K) \equiv_{\Gamma} (G \times_{\Gamma} H) \times_{\Gamma} K$,
   
    \item $G \times_{\Gamma}(H + K) \equiv_{\Gamma} (G \times_{\Gamma} H ) + (G \times_{\Gamma} K)$.
\end{enumerate}

\end{theorem}

\begin{proof}
$(i)$ Consider the mapping $\phi\colon G \times_{\Gamma} H \rightarrow H \times_{\Gamma} G$, such that, 
$\phi(g^{\sigma}, h^{\sigma'}) = (h^{\sigma'}, g^{\sigma}) $
We show that $\phi$ is a $\Gamma$-isomorphism. 
Let $(g_1^{\sigma_1},h_1^{\sigma_1}), (g_2^{\sigma_2},h_2^{\sigma_2}) \in V(G \times_{\Gamma} H)$, suppose, $(g_2^{\sigma_2},h_2^{\sigma_2})$ is a $t$-neighbor of $(g_1^{\sigma_1},h_1^{\sigma_1})$ in $G \times_{\Gamma} H$, which implies, $g_2^{\sigma_2}$ is a $t$-neighbor of $g_1^{\sigma_1}$ in $G$ and $h_2^{\sigma_2}$ is a $t$-neighbor of $h_1^{\sigma_1}$ in $H$, and hence, $(h_2^{\sigma_2},g_2^{\sigma_2})$ is a $t$-neighbor of $(h_1^{\sigma_1},g_1^{\sigma_1})$ in $H \times_{\Gamma} G.$ Thus $\phi$ is a $\Gamma$-isomorphism. 

\bigskip

$(ii)$  Observe that when $\Gamma = \langle e \rangle$,
the function $\phi(g,(h,k)) = ((g,h),k)$
 is a $\langle e \rangle$-isomorphism 
 of $G \times_{\Gamma}(H \times_{\Gamma} K) ~\text{to}~ (G \times_{\Gamma} H) \times_{\Gamma} K$
 where $g \in G$, $h \in H$, and $k \in K$. 
Next we will prove it for general $\Gamma$. 
Notice that by Theorem~\ref{th product existence}, we have, 
\begin{equation}\label{eq product 1}
    \rho_{\Gamma}(G \times_{\Gamma}(H \times_{\Gamma} K)) = \rho_{\Gamma}(G) \times_{\langle e \rangle} \rho_{\Gamma}(H \times_{\Gamma} K) = 
\rho_{\Gamma}(G) \times_{\langle e \rangle} (\rho_{\Gamma}(H) \times_{\langle e \rangle} \rho_{\Gamma}(K))
\end{equation}
and
\begin{equation}\label{eq product 2}
    \rho_{\Gamma}((G \times_{\Gamma} H) \times_{\Gamma} K) = \rho_{\Gamma}(G \times_{\Gamma} H) \times_{\langle e \rangle} \rho_{\Gamma}(K) = 
(\rho_{\Gamma}(G) \times_{\langle e \rangle} \rho_{\Gamma}(H)) \times_{\langle e \rangle} \rho_{\Gamma}(K).
\end{equation}

Since we have already proved 
that our equality holds for 
$\Gamma = \langle e \rangle$, we obtain,
$\rho_{\Gamma}(G) \times_{\langle e \rangle} (\rho_{\Gamma}(H) \times_{\langle e \rangle} \rho_{\Gamma}(K)) \equiv_{\langle e \rangle}
(\rho_{\Gamma}(G) \times_{\langle e \rangle} \rho_{\Gamma}(H)) \times_{\langle e \rangle} \rho_{\Gamma}(K).$
Therefore by equations~(\ref{eq product 1}) and~(\ref{eq product 2}) we have 
$$\rho_{\Gamma}(G \times_{\Gamma}(H \times_{\Gamma} K))
\equiv_{\langle e \rangle}
\rho_{\Gamma}((G \times_{\Gamma} H) \times_{\Gamma} K).$$
By Theorem~\ref{th isomorphism} we have 
$$G \times_{\Gamma}(H \times_{\Gamma} K)
\equiv_{\Gamma}
(G \times_{\Gamma} H) \times_{\Gamma} K.$$
This concludes the proof.

\bigskip

$(iii)$   When $\Gamma = \langle e \rangle$ consider the 
the function $\phi(g,x) = (g,x)$
where $g \in G$ and if $x \in (H+K)$. 
However, here 
if $x \in H$, then the image $(g,x) \in G \times_{\Gamma} H$
and 
if $x \in K$, then the image $(g,x) \in G \times_{\Gamma} K$.
Observe that $\phi$ is a $\langle e \rangle$-isomorphism 
 of $G \times_{\Gamma}(H + K)$ to $(G \times_{\Gamma} H ) + (G \times_{\Gamma} K)$. Next we will prove it for general $\Gamma$. Notice that by Theorem~\ref{th product existence}, we have, \begin{equation}\label{eq product 3}
    \rho_{\Gamma}(G \times_{\Gamma}(H + K)) = \rho_{\Gamma}(G) \times_{\langle e \rangle} \rho_{\Gamma}(H + K) = 
\rho_{\Gamma}(G) \times_{\langle e \rangle} (\rho_{\Gamma}(H) + \rho_{\Gamma}(K))
\end{equation}

and

\begin{equation}\label{eq product 4}
\begin{split}
     \rho_{\Gamma}((G \times_{\Gamma} H ) + (G \times_{\Gamma} K)) &= \rho_{\Gamma}(G \times_{\Gamma} H) + \rho_{\Gamma}(G \times_{\Gamma} K)\\ 
     &= (\rho_{\Gamma}(G) \times_{\langle e \rangle} \rho_{\Gamma}(H)) +  (\rho_{\Gamma}(G) \times_{\langle e \rangle} \rho_{\Gamma}(K)).
\end{split}
  \end{equation}

Since we have already proved 
that our equality holds for 
$\Gamma = \langle e \rangle$, we obtain, 
$$\rho_{\Gamma}(G) \times_{\langle e \rangle} (\rho_{\Gamma}(H) + \rho_{\Gamma}(K)) \equiv_{\langle e \rangle}
(\rho_{\Gamma}(G) \times_{\langle e \rangle} \rho_{\Gamma}(H)) +  (\rho_{\Gamma}(G) \times_{\langle e \rangle} \rho_{\Gamma}(K)).$$ 
Therefore by equations~(\ref{eq product 3}) and~(\ref{eq product 4}) we have 
$$\rho_{\Gamma}(G \times_{\Gamma}(H + K))
\equiv_{\langle e \rangle}
\rho_{\Gamma}((G \times_{\Gamma} H ) + (G \times_{\Gamma} K)).$$
By Theorem~\ref{th isomorphism} we have 
$$G \times_{\Gamma}(H + K)
\equiv_{\Gamma}
(G \times_{\Gamma} H ) + (G \times_{\Gamma} K).$$
This concludes the proof. 
\end{proof}

\begin{remark}
The existence of product and co-product in the category of $(n,m)$-graphs with $\Gamma$-homomorphism playing the role of morphism shows the richness of the category. Moreover, it also shows that the $\Gamma$-homomorphism order, that is an order defined by $G \preceq H$ when $G \xrightarrow{\Gamma} H$ defines a lattice on the set of all $\Gamma$-cores. Here the (cores of the) categorical products and coproducts play the roles of meet and join, respectively. 
\end{remark}

\section{Chromatic number}\label{sec chromatic number}
We know that the ordinary chromatic number of a simple graph $G$ 
can be expressed as the minimum $|V(H)|$ such that $G$ admits a homomorphism to $H$. The analogue of this definition is a popular way for defining chromatic number of other types of graphs, namely, oriented graphs, $k$-edge-colored graphs, $(n,m)$-graphs, signed graphs, push graphs, etc. Here also, we can follow the same. 
The \textit{$\Gamma$-chromatic number} of an $(n,m)$-graph is given by,
$$\chi_{\Gamma: n,m}(G) = \min\{|V(H)| : G\xrightarrow{\Gamma} H\}.$$ 
Moreover, for a family $\mathcal{F}$ of $(n,m)$-graphs, 
the $\Gamma$-chromatic number is given by, 
$$\chi_{\Gamma: n,m}(\mathcal{F}) = \max\{\chi_{\Gamma: n,m}(G) : G \in \mathcal{F}\}.$$

Let $\Gamma \subseteq S_{2n+m}$ be an Abelian group acting on the set $A_{n,m}.$ For $x \in A_{n,m}$, we call the set, $\text{Orb}_x = \{ \sigma(x) : \sigma \in \Gamma \}$ an \textit{orbit of $x$}. A \textit{consistent} group $\Gamma \subset S_{2n+m}$ is such that each orbit induced by $\Gamma$ acting on the set $A_{n,m}$ contains $\overline{i}$ if and only if it contains $i$ for $i \in \{1, 2, \ldots, 2n \}$. 
Notice that, these orbits form a partition on the set $A_{n,m}$ as the relation, $x \sim y$ whenever $x = \sigma(y)$ for some $\sigma \in \Gamma$, is an equivalence relation. 

Next we establish a useful observation. 

\begin{theorem}
Let $G$ be any $(n,m)$-graph, $\Gamma$ be a switch-commutative group, then we have, $$\chi_{\Gamma: n,m}(G) \leq \chi_{n,m}(G) \leq |\Gamma|. \chi_{\Gamma: n,m}(G).$$
\end{theorem}

\begin{proof}
    As $\Gamma$-homomorphism is, in particular, an $\langle e \rangle$-homomorphism, the first inequality holds. The second inequality follows from Proposition~\ref{prop chromatic num}.
\end{proof}

An immediate corollary follows. 

\begin{corollary}
    Let $G$ be an $(n,m)$-graph and let $\Gamma_1, \Gamma_2 \subseteq S_{2n+m}$ be two switch-commutative groups.  Then we have 
    $$\frac{\chi_{\Gamma_1 : n,m}(G)}{|\Gamma_2|} \leq \chi_{\Gamma_2 : n,m}(G) \leq |\Gamma_1| \cdot \chi_{\Gamma_1: n,m} (G).$$
   \end{corollary}

\begin{proposition}\label{obs orbit}
   Let $\Gamma \subseteq S_{2n+m}$ be a consistent group, $G$ be an $(n,m)$-graph, and  $G'$ be a $\Gamma$-equivalent graph of $G$. If a vertex $v$ is a $t$-neighbor of $u$ in $G$,
   then $v$ must be a $\sigma(t)$ neighbor of $u$ in $G'$ for some $\sigma \in \Gamma$. 
\end{proposition}

\begin{proof}
Observe that, it is enough to prove the statement assuming  $G'$ is obtained from $G$ by performing a $\sigma$-switch on $v$.

By Lemma~\ref{lem convention}, in $G'$, $v$ is a $\overline{\sigma(\overline{t})}$-neighbor of $u$. Note that, $\overline{t} \in Orb_t$ as $\Gamma$ is a consistent group. 
Since $\overline{t} \in Orb_t$, we have $\sigma(\overline{t}) \in Orb_{\overline{t}} = Orb_t$. Therefore, 
$\overline{\sigma(\overline{t})} \in Orb_t$ as $\Gamma$ is consistent. 
\end{proof}

Next we focus on studying the $\Gamma$-chromatic number of $(n,m)$-forests.

\begin{theorem}
Let $\mathcal{F}$ be the family of $(n,m)$-forests and let $k$ be the number of orbits   
 of $A_{n,m}$  with respect to the action of $\Gamma$. 
Then,
\[   \chi_{\Gamma : n,m}(\mathcal{F})  \leq \begin{cases} 
		
		k+2  &~\text{if}~ $k$ ~\text{is even}, \\
		k+1 &~\text{if}~ $k$ ~\text{is odd}.
	\end{cases} \]
	Moreover, equality holds if $\Gamma$ is consistent. 
\end{theorem}

\begin{proof}
 We first prove the upper bound. Assume that $k$ is odd. In this case consider the complete graph $K_{k+1}$. 
 We will construct a complete $(n,m)$-graph having $K_{k+1}$ as its underlying graph. 
 As $(k+1)$ is even, we know that $K_{k+1}$ 
 can be decomposed into $\frac{k-1}{2}$ edge-disjoint Hamiltonian cycles  and a perfect matching.

  Let 
  $\{ \alpha_1, \alpha_2, \ldots \alpha_k \}$ be the representatives of the $k$ orbits.  Assume that if $\alpha$ and $\bar{\alpha}$ are in different orbits for some $\alpha \in \{2, 4, \ldots, 2n\}$, then either both $\alpha$ and $\bar{\alpha}$ are chosen as representatives, or neither of them are chosen as representatives.

  If $\alpha$ and $\bar{\alpha}$ are both chosen as representatives, 
  they are called \textit{representative pairs} $(\alpha, \bar{\alpha})$.
  For any representative pairs $(\alpha, \bar{\alpha})$,
  notice that $\alpha$ must be color of an arc. For each representative pair $(\alpha, \bar{\alpha})$ 
    orient one of the Hamiltonian cycles as a directed cycle, and color all its arcs with $\alpha$ (equivalently, all its reverse arcs with 
  $\bar{\alpha}$). That means,  every vertex of the graph will have an $\alpha$-neighbor and a $\bar{\alpha}$-neighbor.

  For those representatives  $\alpha_i$ for which $\bar{\alpha}$ also belongs to the same orbit, we will arbitrarily choose a previously not chosen $\alpha_j$ that is not part of any representative pair, where 
  $i, j \in \{1, 2, \ldots, k\}$. Now we will choose a Hamiltonian cycle which is not colored, and will color its edges alternatively with the colors $\alpha_i$ and $\alpha_j$ while traversing it in a clockwise direction (with respect to any embedding of the cycle). This will ensure that every vertex of the graph has a $\beta_i$-neighbor and a $\beta_j$-neighbor, where 
  $\beta_i \in \{\alpha_i, \bar{\alpha_i}\}$ and 
  $\beta_j \in \{\alpha_j, \bar{\alpha_j}\}$. 

  After finishing this process, as $k$ is odd, one perfect matching will remain uncolored and one $\alpha_l$ will remain unused (in the above process). This $\alpha_l$ does not belong to a representative pair, and thus $\bar{\alpha}$ belongs to the orbit of $\alpha_l$. 
  Finally we will color the edges of the perfect matching with $\alpha_l$. 
  Thus all vertices of the graph has either an $\alpha_l$-neighbor, 
  or a $\bar{\alpha_l}$-neighbor. 
With a little abuse of notation, we denote the so obtained $(n,m)$-graph by $K_{k+1}$ itself. 
   
   \medskip

  We claim that every $(n,m)$-forest admit a $\Gamma$-homomorphism to $K_{k+1}$. If not, then 
   there exists a minimal (with respect to number of vertices) counter-example  $F$ 
   that does not admit $\Gamma$-homomorphism to $K_{k+1}$.
   Let $u$ be a leaf with the vertex $v$ as its only neighbor in $F$, then  $F \setminus \{u\}$ is no 
   longer a minimal counter-example, thus it admits a $\Gamma$-homomorphism $f$ to $K_{k+1}$.
  That means, there exists a $\Gamma$-equivalent $(n,m)$-graph 
  $F'$ of $F \setminus \{u\}$ such that $f$ is a 
  $\langle e \rangle$-homomorphism of $F'$ to $K_{k+1}$. 
  Suppose $v$ is a $\beta$-neighbor of $u$, and  $\beta$ belongs to an orbit whose representative is $\alpha$. If $\alpha$ belongs to a representative pair, 
  then switch $u$ to convert $v$ into an $\alpha$-neighbor of $u$. 
  Otherwise, switch $u$ to convert $v$ into a 
  $\bar{\alpha}$-neighbor of $u$. 
 Since $u$ is a leaf,  this does not affect any adjacencies in $F'$.  
    In either case, $u$ is an $\alpha_i$-neighbor of $v$ for some $i \in \{1, 2, \ldots, k\}$. 

   Now, we extend $f$ to an
   $\langle e \rangle$-homomorphism of $F'$ to $K_{k+1}$
   by mapping $u$ to the $\alpha_i$-neighbor $f(v)$ in $K_{k+1}$, since every vertex in $K_{k+1}$ has an $\alpha_i$-neighbor for every $i \in \{1,2, \ldots, k\}$. 
   That means, there exists a $\Gamma$-homomorphism of $F$ to $K_{k+1}$.
    This contradicts the minimality of $F$. Hence every $(n,m)$-forest admits a $\Gamma$-homomorphism to $K_{n,m}$.

  \medskip
  
  Secondly, assume that $k$ is even. Note that, if there were $(k+1)$ orbits instead, then by what we have proved above, it would be possible to show that 
  every $(n,m)$-forest admits a $\Gamma$-homomorphism to an $(n,m)$-graph having
  $K_{k+2}$ as underlying graph. Therefore, assuming a dummy orbit we are done with this case too.

  \medskip

  Next we will prove the tightness of the upper bound when $\Gamma$ is consistent. 
  Let $\{ \alpha_1, \alpha_2, \ldots ,\alpha_k \}$ be the representatives of the $k$ orbits.   For odd values of $k$, 
   consider the star $(n,m)$-graph $S$ on $(k+1)$ vertices:
   the central vertex $v$ having $k$ neighbors $v_1, v_2, \ldots, v_k$.  
   Let $v_i$ be a $\alpha_i$-neighbor of $v$ for all $i \in \{1,2, \ldots, k\}$.
    No matter how we switch the vertices of $S$, the vertex $v$ will have
     $k$ distinctly adjacent neighbors. 
     Therefore we have $\chi_{\Gamma : n,m}(S) \geq k+1$, and thus  
     $\chi_{\Gamma : n,m}(\mathcal{F}) = k+1$ when $k$ is odd and $\Gamma$ is consistent.

For even values of $k$, consider a rooted tree $T$ of height two in which every vertex, other than the leaves, has exactly one $\alpha_i$-neighbor for $i \in \{1,2, \ldots, k\}$. 
Suppose $T$ admits a $\Gamma$-homomorphism $f$ to an $(n,m)$-graph $H$. 
Let $r$ be the root of $T$. If $H$ has $(k+1)$ vertices, then the images of
the vertices from $N[r]$ 
under $f$ will be a spanning subgraph in $H$. 
Furthermore, notice that each vertex of $N[r]$ has at least one $\beta_i$-neighbor, where $\beta_i$ belongs to the $i^{th}$ orbit. 
Thus their images should also have the same property, that is each of them must have at least one $\beta_i$-neighbor, where $\beta_i$ belongs to the $i^{th}$ orbit. However, as $H$ has only $(k+1)$ vertices, each of its vertices are forced to have exactly one 
$\beta_i$-neighbor, where $\beta_i$ belongs to the $i^{th}$ orbit. Now if we restrict ourselves to only the neighbors whose type is from a particular orbit, that must give us a perfect matching, which is impossible as $(k+1)$ is odd. Therefore, $H$ must have at least $(k+2)$ vertices which implies the lower bound. 
\end{proof}

The above result implies the upper bound of Theorem~1.1 of~\cite{nevsetvril2000colored}. 

\section{Concluding remarks}\label{sec conclusion}
In this article, we introduced a generalized switch operation on $(n,m)$-graphs and studied their basic algebraic properties. Naturally, this topic generates a lot of interesting open questions, especially, in an effort of extending the known results in the domain of graph homomorphisms. We list a few of them here. 
\begin{enumerate}[(i)]
    \item Is it possible to generalize the notion of exponential graphs (see section $2.4$ in \cite{hell2004graphs}) 
    in the category of $(n,m)$-graphs with respect to $\Gamma$-homomorphism 
    where $\Gamma$ is switch-commutative? 
    
    \item Is it possible to obtain the analogue of the Density Theorem (see Theorem $3.30$ in ~\cite{hell2004graphs}) in this context? At present, such an analogue is unknown even for $(0,2)$-graphs.

    \item Finding the $\Gamma$-chromatic number for other graph families like planar graphs, partial $2$-trees, outerplanar graphs, cycles, grids, graphs having bounded maximum degree, etc. can be other directions of study in this set up.

    \item Characterizing $(n,m)$-graphs  $\Gamma$-equivalent to a $(1,0)$-graph (that is, monochromatic) is an interesting natural problem.

    \item Studying the $(n,m)$-cycles may further lead us to finding a characterization of $\Gamma$-equivalent graphs, similar to what Zaslavsky~\cite{zaslavsky1982signed} did for signed graphs. 
    
    \item If $\Gamma$ is not switch-commutative, then does the categorical product exist? 

    \item Given a pre-decided $(n,m)$-graph $H$, a switch-commutative group $\Gamma$ and an input $(n,m)$-graph $G$, 
    the decision problem 
    ``does $G$ admit a $\Gamma$-homomorphism to $H$?'' is in $NP$ due to Proposition~\ref{prop chromatic num}. Can we characterize the full dichotomy of this problem? 
\end{enumerate}

As a remark, it is worth mentioning that using the notion of generalized switch (implicitly), it was possible to improve the existing upper bounds of the $\langle e \rangle$-chromatic number of $(n,m)$-partial $2$-trees where $2n+m = 3$~\cite{lahiri2021chromatic}. Therefore, it will not be surprising if $\Gamma$-homomorphism becomes useful  as a technique to establish bounds for $(n,m)$-chromatic number of graphs. 

\acknowledgements
\label{sec:ack}
We would like to thank the anonymous reviewers. Their comments and suggestions have greatly contributed to the improvement of the manuscript.

\nocite{*}
\bibliographystyle{abbrvnat}
% use the following instead if you encounter problems 
%\bibliographystyle{alpha}
\bibliography{sample-dmtcs}
\label{sec:biblio}

\end{document}